\documentclass{vldb}
\usepackage{etoolbox}
\newtoggle{expand}
\toggletrue{expand}
\usepackage{paralist}
\usepackage{graphicx}
\usepackage{balance}  % for  \balance command ON LAST PAGE  (only there!)

\usepackage{url}
\usepackage{framed}
\usepackage[skip=5pt,tableposition=top,figureposition=bottom,font={small},]{caption}
\usepackage[skip=0pt]{subcaption}

\usepackage[noend]{algpseudocode}
\usepackage{algorithm}
\usepackage{color,colortbl}

\usepackage{multirow}
\usepackage{rotating}
\usepackage{epstopdf}
\usepackage{pgfplots}
\pgfplotsset{compat=newest}  %removed because otherwise I cannot compile -- JH
\usepackage{tikz}
\usetikzlibrary{arrows}
\usetikzlibrary{trees}
\usetikzlibrary{pgfplots.groupplots}
\mathchardef\mhyphen="2D

\usepackage{mdwlist}

\definecolor{lightgray}{gray}{0.8}

\def\ojoin{\setbox0=\hbox{$\bowtie$}%
  \rule[-.02ex]{.25em}{.4pt}\llap{\rule[\ht0]{.25em}{.4pt}}}
\def\leftouterjoin{\mathbin{\ojoin\mkern-5.8mu\bowtie}}

\newcommand{\nodeLabel}[1]{\ensuremath{\ifthenelse{\equal{#1}{}}{\lambda_N}{\lambda_N(#1)}}}
\newcommand{\edgeLabel}[1]{\ensuremath{\ifthenelse{\equal{#1}{}}{\lambda_E}{\lambda_E(#1)}}}
\newcommand{\graph}{\ensuremath{G=\langle N, E, \allowbreak \nodeLabel{}, \allowbreak \edgeLabel{}\rangle}}
\newcommand{\pid}[2]{\ensuremath{\mathit{pId}_{#1}(\mathit{#2})}}
\newcommand{\pidname}[1]{\ensuremath{\mathit{pId}_{#1}}}
\newcommand{\sig}[2]{\ensuremath{\mathit{sig}_{#1}(\mathit{#2})}}
\newcommand{\signame}[1]{\ensuremath{sig_{#1}}}
\newcommand{\kbisim}{$k$-\emph{bisimulation}}
\newcommand{\kbisimnormal}{$k$-bisimulation}
\newcommand{\s}{\emph{S}}
\newcommand{\pq}{\emph{pQueue}}
\newcommand{\nt}{\ensuremath{N_t}}
\newcommand{\et}{\ensuremath{E_t}}
\newcommand{\nodetable}{\nt}
\newcommand{\edgetable}{\et}

\newcommand{\ignore}[1]{}

\newcommand{\construct}{\textsc{Build\_Bisim()}}

\newcommand{\addnodes}{\textsc{Add\_Nodes()}}

\newcommand{\addedges}{\textsc{Add\_Edges()}}

\newcommand{\ptkbisim}[1]{\ensuremath{\stackrel{*}{\approx}_{#1}}} %paige tarjan
\newcommand{\mykbisim}[1]{\ensuremath{\approx^{#1}}} %our version

\newcommand{\pidt}[2]{\ensuremath{\mathit{pId_{#1\_#2}}}} %pids for table

\newlength{\leftbarwidth}
\newlength{\leftbarsep}
\colorlet{leftbarcolor}{black}
\renewenvironment{leftbar}{%
    \vspace{-5pt}
    \MakeFramed {\advance \hsize -\width \FrameRestore }%
}{%
    \endMakeFramed
}

\newtheorem{theorem}{Theorem}
\newtheorem{definition}{Definition}
\newtheorem{proposition}{Proposition}
\begin{document}

% ****************** TITLE ****************************************

%\title{I/O-efficient algorithms for localized bisimulation partition
%construction and maintenance on massive graphs}
\title{External memory bisimulation reduction of big graphs}
% possible, but not really needed or used for PVLDB:
%\subtitle{[Extended Abstract]
%\titlenote{A full version of this paper is available as\textit{Author's Guide to Preparing ACM SIG Proceedings Using \LaTeX$2_\epsilon$\ and BibTeX} at \texttt{www.acm.org/eaddress.htm}}}

% ****************** AUTHORS **************************************

%\iftoggle{expand}{

\numberofauthors{5} %  in this sample file, there are a *total*

\author{
% 1st. author
\alignauthor
Yongming Luo\\
       \affaddr{Eindhoven University of Technology}\\
       \affaddr{The Netherlands}\\
       \tt{y.luo@tue.nl}
% 2nd. author
\alignauthor
George H.\,L.\ Fletcher\\
       \affaddr{Eindhoven University of Technology}\\
       \affaddr{The Netherlands}\\
       \tt{g.h.l.fletcher@tue.nl}
% 3rd. author
\alignauthor
Jan Hidders\\
       \affaddr{Delft University of Technology}\\
       \affaddr{The Netherlands}\\
       \tt{a.j.h.hidders@tudelft.nl}
\and  % use '\and' if you need 'another row' of author names
% 4th. author
\alignauthor
Yuqing Wu\\
       \affaddr{Indiana University, Bloomington}\\
       \affaddr{USA}\\
       \tt{yuqwu@cs.indiana.edu}
% 5th. author
\alignauthor
Paul De Bra\\
       \affaddr{Eindhoven University of Technology}\\
       \affaddr{The Netherlands}\\
       \tt{P.M.E.d.Bra@tue.nl}
}
%\additionalauthors{}
%}{
%do nothing in shorter version
%}

\maketitle

\begin{abstract}

In this paper, we present, to our knowledge, the first known I/O efficient
solutions for computing the \kbisimnormal\ partition of a massive directed
graph, and performing maintenance of such a partition upon updates to the
underlying graph.  Ubiquitous in the theory and application of graph data,
bisimulation is a robust notion of node equivalence which intuitively groups
together nodes in a graph which share fundamental structural features.
\kbisimnormal\ is the standard variant of bisimulation where the topological
features of nodes are only considered within a local neighborhood of radius
$k\geqslant 0$.

The I/O cost of our partition construction algorithm is bounded
by $O(k\cdot \mathit{sort}(|\et|) + k\cdot scan(|\nt|) + \mathit{sort}(|\nt|))$, while our
maintenance algorithms are bounded by $O(k\cdot \mathit{sort}(|\et|) + k\cdot
\mathit{sort}(|\nt|))$.
The space complexity bounds are $O(|\nt|+|\et|)$ and $O(k\cdot|\nt|+k\cdot|\et|)$, resp.
Here,  $|\et|$ and $|\nt|$ are the number of disk pages occupied
by the input graph's edge set and node set, resp., and $\mathit{sort}(n)$ and $\mathit{scan}(n)$ are the cost
of sorting and scanning, resp.,  a file occupying $n$ pages in external memory.
Empirical analysis on a variety of massive real-world and synthetic graph
datasets shows that our algorithms perform efficiently in practice, scaling
gracefully as graphs grow in size.

%The source code of the algorithms will be made publicly available, inviting further discussion as well as optimization upon our solutions.
\end{abstract}

%\showthe\dbltextfloatsep
%\showthe\dblfloatsep
%\showthe\abovecaptionskip
%\showthe\belowcaptionskip

\section{Introduction}\label{sec:intro}

Massive graph-structured datasets are becoming increasingly common in a wide
range of applications.  Examples such as social
networks, linked open data, and
biological networks have drawn much attention in both industry
and academic research.  In reasoning over graphs, a fundamental and ubiquitous
notion is that of bisimulation,  which is a characterization of when two nodes
in a graph share basic structural properties such as neighborhood
connectivity. Bisimulation arises and is widely adopted in a surprisingly large
range of research fields~\cite{Sangiorgi:2011:ATB:2103601}.  In data management,
bisimulation partitioning (i.e., grouping together bisimilar
nodes in order to reduce graph size) is often a basic step in indexing semi-structured datasets~\cite{Milo1999}, and also finds fundamental applications in RDF
\cite{saintdb2012} and general graph data (e.g.,
compression~\cite{Buneman:2003:PQC:1315451.1315465,Fan:2012:QPG:2213836.2213855},
query processing~\cite{Kaushik2002}, data analytics~\cite{Fan:2012:GPM:2274576.2274578,Tian2008}).

It is often the case that bisimulation reductions of real graphs
result in partitions which are
too refined for effective use.  Hence, a notion of localized
\kbisimnormal\ has proven to be quite successful in data management applications
(e.g., \cite{FletcherGWGBP09,Kaushik2002,QunLO03,YiHSY04}).  \kbisimnormal\ is the
variant of bisimulation where topological features of nodes are only
considered within a local neighborhood of radius $k\geqslant 0$.  With a
pay-as-you-go nature, \kbisimnormal\ is cheaper to compute and maintain, cost
adjustable, and faithfully representative of the bisimulation partition within the
local neighborhood.

\subsection*{State of the art}

Algorithms for bisimulation partitioning have been studied for
decades,  with well-known algorithms such as those of Paige and Tarjan
\cite{paige1987three} and more recent work (e.g., \cite{Dovier2004}), having
effective theoretical behavior.  

In practice, however, state-of-the-art solutions face a critical challenge:
all known approaches for computing bisimulation are internal-memory
based solutions.\footnote{With the single exception of Hellings et al.\
\cite{Hellings:2012:EEB:2213836.2213899} which we discuss below in Section 
\ref{sec:more_discussion}.} As such, their
inherently random memory access patterns do not translate to efficient I/O-bound
solutions, where it is crucial to avoid such access patterns.  Consequently,
when processing graphs which do not fit entirely in main memory 
the performance of these algorithms decreases drastically.

The reality is that, in practice, many graphs of interest are too large to be
processed in main memory.  Indeed, massive graphs are now ubiquitous
\cite{Fan:2012:GPM:2274576.2274578,Heath2011}.  Furthermore, the size of graphs
will only continue to grow as technologies for generating and capturing data
continue to improve and proliferate.  We can safely conclude that it will become
increasingly infeasible to apply existing internal-memory bisimulation partition
algorithms in practice.

To process real graphs, therefore, we must necessarily turn to either external
memory, distributed, or parallel solutions.  There has been some work on
parallel (e.g.,~\cite{Rajasekaran98,smolka95}) and distributed
(e.g.,~\cite{distributebisim}) approaches to bisimulation computation, and,
recently, external memory solutions on restricted acyclic and tree-structured
graphs~\cite{Hellings:2012:EEB:2213836.2213899}. However, to our knowledge
there is no known effective solution for computing bisimulation and \kbisimnormal\ partitions on arbitrary
graph structures in external memory.  Such an algorithm would not only enable
us to process big graphs on single machines, but also provide an essential step for
parallel and distributed solutions (e.g., MapReduce) to further scale their
performance on real graphs.

\subsection*{Our contributions}

Given these motivations, we have studied external memory solutions for reasoning
about \kbisimnormal\ on arbitrary graphs.  In this paper, we present the results
of our study, which makes the following high-level contributions.

\begin{compactitem}

  \item We present \ignore{, to our knowledge, }the first known I/O efficient external
      memory based algorithm for constructing the \kbisimnormal\ partition of a
      disk-resident graph.
The I/O cost of this algorithm is bounded
by $O(k\cdot \mathit{sort}(|\et|) + k\cdot \mathit{scan}(|\nt|) + \mathit{sort}(|\nt|))$,
with space complexity $O(|\nt|+|\et|)$,
where $|\et|$ and $|\nt|$ are the number of disk pages occupied
by the input graph's edge set and node set, resp., and $\mathit{sort}(n)$ and $\mathit{scan}(n)$ are the cost
of sorting and scanning, resp.,  a file occupying $n$ pages in external memory.

\item We present \ignore{, to our knowledge, }the first known I/O efficient
    external memory based algorithms for performing maintenance on
a disk-resident \kbisimnormal\ graph partition, with
I/O cost bounded by $O(k\cdot \mathit{sort}(|\et|) + k\cdot \mathit{sort}(|\nt|))$,
and space complexity $O(k\cdot|\nt|+k\cdot|\et|)$.

  \item We present the results of an extensive empirical analysis of our
      solutions on a variety of massive real-world and synthetic graph datasets,
      showing that our algorithms not only perform efficiently, but also scale
      gracefully as graphs grow in size.  For example, the 10-bisimulation
      partition of a graph having 1.4 billion edges can be computed with our
      solution within a day on commodity hardware, while this would take weeks,
      if not months, for a traditional in-memory algorithm to accomplish in the
      same environment.

\end{compactitem}
%We note that parallel and distributed solutions could also benefit from the
%basic novel ideas developed
%here;  we leave such explorations open as avenues for future
%research.
The rest of the paper is organized as follows. In the next section we give our
basic definitions and data structures used. We then describe in Section
\ref{sec:construction}  our solution for constructing localized bisimulation
partition.  Next, Section \ref{sec:maint} presents algorithms for keeping an
existing partition up to date, in the face of  updates to the underlying graph.
Section \ref{section:experiment} presents the results of our empirical study of
all algorithms.  We then conclude in Section \ref{sec:conclude} with a
discussion of future directions for research.

\section{Preliminaries}\label{sec:prelims}

\subsection{Data model and definitions}
Our data model is that of finite directed node- and
edge-labeled graphs
$\langle N, E, \allowbreak \nodeLabel{}, \allowbreak \edgeLabel{}\rangle$,
 where $N$ is a finite set of nodes, $E\subseteq N\times
N$ is a set of edges,
$\nodeLabel{}$ is a function from $N$ to a set of node labels $\mathcal{L}_N$, and
$\edgeLabel{}$ is a function from $E$ to a set of edge labels $\mathcal{L}_E$.

\begin{definition}
\label{def:kbisim}
%Let $k\geq 0$ and \graph\ be a graph.
Let $k$ be a non-negative integer and \graph\ be a graph.
Nodes $u,v \in N$ are called $k$-\emph{bisimilar} (denoted as $u \approx^k v$),
iff the following holds:
\begin{compactenum}
  \item
  $\nodeLabel{u} = \nodeLabel{v}$,
  \item
  if $k>0$, then $\forall u' \in N [(u,u') \in E \Rightarrow \exists v' \in N [(v,v') \in E,~u' \approx^{k-1} v'\textrm{~and~}\edgeLabel{u,u'} = \edgeLabel{v,v'}\footnotemark[1]] ]$, and
  \item
  if $k>0$, then $\forall v' \in N [(v,v') \in E \Rightarrow \exists u' \in N [(u,u') \in E,~v' \approx^{k-1} u'\textrm{~and~}\edgeLabel{v,v'} = \edgeLabel{u,u'}] ]$.
\end{compactenum}
\end{definition}

\footnotetext[1]{Note that we use \edgeLabel{u,u'}, instead of
\edgeLabel{(u,u')}, for ease of readability.}

It can be easily shown that the $k$-\emph{bisimilar} relation is an equivalence relation.

We illustrate Definition \ref{def:kbisim} with an example.  Consider the graph
given in Figure \ref{fig:example_graph}.  It is a small social network graph, in which nodes 1 and 2 are
0- and
1- bisimilar but not 2-bisimilar.

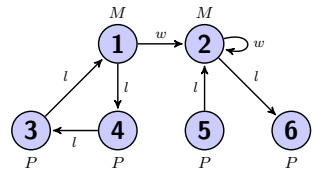
\begin{figure}[htbp]
\centering
%\fbox{
\resizebox{.5\columnwidth}{!}{
\begin{tikzpicture}[->,>=stealth',shorten >=1pt,auto,node distance=1.7cm,
  thick,main node/.style={circle,fill=blue!20,draw,font=\sffamily\Large\bfseries}]
\centering
\useasboundingbox (-2,-2) rectangle (4,0.5);
  \node[main node] (1) [label=above:$M$] {1};
  \node[main node] (2) [label=above:$M$,right of=1] {2};
  \node[main node] (4) [label=below:$P$,below of=1] {4};
  \node[main node] (3) [label=below:$P$,left of=4] {3};
  \node[main node] (5) [label=below:$P$,below of=2] {5};
  \node[main node] (6) [label=below:$P$,right of=5] {6};
  \path[every node/.style={font=\sffamily\small}]
    (1) edge node {$w$} (2)
    (1) edge node {$l$} (4)
    (2) edge [loop right] node {$w$} (2)
    (2) edge node {$l$} (6)
    (3) edge node {$l$} (1)
    (4) edge node {$l$} (3)
    (5) edge node {$l$} (2);
%\draw (current bounding box.south west) rectangle (current bounding box.north east);
\end{tikzpicture}
}
%}
%\vspace{-10pt}
\caption{Example graph of a social network, where nodes 1 and 2 have label \emph{M} (short for ``manager''), and the other nodes
have label \emph{P} (short for ``people'').  The edge label \emph{l} is short
for ``likes'', while \emph{w} is short for ``works for''. }
\label{fig:example_graph}
\end{figure}

Recall from Section \ref{sec:intro} that our interest in this paper is in
computing the \kbisimnormal\ partition of a massive graph, and performing maintenance
on the result under updates to the original graph.  By {\em massive}, we mean
that both the set of nodes and the set of edges of the graph are too big to fit
into main memory.  By a {\em partition} of the graph, we mean an assignment of
each node $u$ of the graph to a {\em partition block}, which is the
unique subset of nodes in the graph of which the members are $k$-bisimilar to $u$.

In particular, we are interested in constructing partition ``identifiers.''

\begin{definition}
A {\em $k$-partition identifier} for graph \graph\ and $k\geq 0$ is a set of $k+1$ functions
$\mathcal{P} = \{\pidname{0}, \ldots, \pidname{k}\}$ such that, for each $0\leq i
\leq k$,
$\pidname{i}$ is a function from
$N$ to the integers, and,  for all nodes $u,v \in N$, it holds
that \pid{i}{u} = \pid{i}{v} iff $u\approx^i v$.
\end{definition}

A fundamental tool in our reasoning about \kbisimnormal\ is the notion of node
signatures.

\begin{definition}
\label{def:sig}
Let \graph\ be a graph, $k\geq 0$, and
$\mathcal{P} = \{\pidname{0}, \ldots, \pidname{k}\}$ be a $k$-partition
identifier for $G$.
The $k$ {\em bisimulation signature} of node $u \in N$
is the pair \sig{k}{u} $= (\pid{0}{u}, L)$ where:
\[
L=
\begin{cases}
\emptyset & \text{if } k=0, \\
\{(\edgeLabel{u,u'},\pid{k-1}{u'})\mid (u,u') \in E\} & \text{if } k>0.
\end{cases}
\]
\end{definition}

We then have the following fact.
\begin{proposition}
\label{prop:equal}
$\pid{k}{u}=\pid{k}{v}$ iff $\sig{k}{u}=\sig{k}{v}\allowbreak (k\ge0)$.
\end{proposition}
\iftoggle{expand}{
A proof of Proposition \ref{prop:equal} can be found in Appendix \ref{proof:equal}.
}{
Proof omitted.
%do nothing in shorter version
}

Proposition \ref{prop:equal} is the basis of all algorithms in this paper.  The
basic idea is that a node's \kbisimnormal\ partition block can be determined by
its \kbisimnormal\ signature, which in turn is determined by the
$(k-1)$-bisimulation partition of the graph.  Intuitively, in order to compute
the \kbisimnormal\ partition, we  compute the graph's $j$-\emph{bisimulation}
($0\le j \le k$) partitions bottom-up, starting from $j=0$.  We call each
such intermediate computation the \emph{iteration $j$ computation}.

It is straightforward to show that the $k$-bisimulation partition of a graph is
unique.  Hence, in the sequel, we can safely talk about $k$-partition
identifiers as unique objects.  Also, note that we will use integer node
identifier values (denoted as $uId$) to designate nodes $u \in N$. Therefore, in
the following discussions the functions \signame{k} and \pidname{k} both could take
node identifiers (i.e., integers) as input.

%construct signatures and partition identifiers in a bottom-up approach (from lower $k$ to higher $k$).

\begin{table}[htbp]
\centering
\caption{\kbisim\ for the example graph in Figure \ref{fig:example_graph} ($k=0,1,2$)}
\label{table:example_table}
%\singlespacing
\footnotesize
\tabcolsep=0.1cm
\resizebox{\columnwidth}{!}{
\begin{tabular}{|c|c|c|c|c|c|}
\hline
\emph{nId} & \pid{0}{nId} & \sig{1}{nId} & \pid{1}{nId} & \sig{2}{nId} & \pid{2}{nId}\\
\hline
1 & 1 & $1,\{(w,1),(l,2)\}$   & 3 & $1,\{(w,3),(l,5)\}$ & 7\\
2 & 1 & $1,\{(w,1),(l,2)\}$   & 3 & $1,\{(w,3),(l,6)\}$ & 8\\
3 & 2 & $2,\{(l,1)\}$        & 4 & $2,\{(l,3)\}$      & 9\\
4 & 2 & $2,\{(l,2)\}$        & 5 & $2,\{(l,4)\}$      & 10\\
5 & 2 & $2,\{(l,1)\}$        & 4 & $2,\{(l,3)\}$      & 9\\
6 & 2 & $2,\{\}$            & 6 & $2,\{\}$          & 11\\
\hline
\end{tabular}
}
\end{table}

Table \ref{table:example_table} shows one way of assigning \kbisimnormal\ ($k=0,1,2$)
partition identifiers and signatures for the example graph in Figure
\ref{fig:example_graph}, where the \emph{nId} denotes the unique identifier for each node,
and \pid{i}{nId} and \sig{j}{nId} ($0~\le~i~\le~2$ and  $0~<~j~\le~2$) are presented accordingly.
%Here shows the intermediate results if we apply Algorithm \ref{algo:build_k-bisims} on
%For each node \emph{u}, we present .
%Note that here we assume the numbering scheme for \s\ is a self-increased counter across iterations.
For $k=0$, nodes are grouped into two partitions by node labels (given
identifiers 1 and 2).
Then for $k=1,2$, signatures are constructed according to Definition
\ref{def:sig}, and then distinct partition identifiers are assigned to distinct
signatures, following Proposition \ref{prop:equal}.

\ignore{
\subsection{The problem statement}

\kbisim\ is a localized variant of the standard notion of bisimulation, a
ubiquitous structural notion in the study of graphs.  In this paper, we develop,
to our knowledge, the first known efficient approaches to computing the \kbisim\
partition of a massive graph, and performing maintenance on the result under
updates to the original graph.  By {\em massive}, we mean that both the set of
nodes and the set of edges of the graph are too big to fit into main memory.  By
a {\em partition} of the graph, we mean an assignment of each node $u$ of the
graph to a unique {\em partition block}, which is the subset of nodes in the
graph which are $k$ bisimilar to $u$.  By {\em efficient}, we mean that our
algorithms minimize random access on disk as much as possible.  In other words,
our solutions effectively determine the $k$ bisimulation equivalence classes of
a given graph stored in external memory, and efficiently maintain these classes
as the graph evolves.
}

\subsection{Data structures}

We assume that graphs are saved on disk in the form of fixed column tables (node
set as table \nodetable\ and edge set as table \edgetable).  We also assume
that these tables can have several copies sorted on different columns.  In later
discussions, we will use the notation $X.y$ to refer to column $y$ of table $X$.

We have the following possible attributes for \nodetable:

\vspace{2pt}
\noindent
\begin{tabular}{| p{.17\columnwidth}|p{.75\columnwidth}|}
\hline
%\boldmath
\emph{nId} & node identifier (note that this is the same as row identifier in the table;  we leave this attribute here for clarity of the discussion). \\
\hline
\emph{nLabel} & node label\\
\hline
\pidt{old}{nId} & \emph{bisimulation} partition identifier for the given \emph{nId} from last computation iteration\\
\hline
\pidt{new}{nId} & \emph{bisimulation} partition identifier for the given \emph{nId} from the current computation iteration\\
\hline
\pidt{j}{nId} & $j$ \emph{bisimulation} partition identifier for the given \emph{nId} ($j=0,1,\dots,k$)\\
\hline
\end{tabular}
\vspace{0pt}

\noindent
and for \edgetable:

\vspace{2pt}
\noindent
\begin{tabular}{ |p{.17\columnwidth} | p{.75\columnwidth}| }
\hline
\emph{sId} & source node identifier \\
\hline
\emph{tId} & target node identifier \\
\hline
\emph{eLabel} & edge label \\
\hline
\pidt{old}{tId} & \emph{bisimulation} partition identifier for the given \emph{tId} from last computation iteration \\
\hline
\end{tabular}
\vspace{0pt}

We further assume that we have a {\em signature storage facility} \s, which
stores the mapping between signatures and their corresponding partition
identifiers.  \s\ is a data structure having only one idempotent function called
\emph{S.insert()}.  For node $u\in N$, \emph{S.insert()} takes \sig{j}{u} $(0\le j \le k)$ as
input, and provides \pid{j}{u} as output.  Essentially \emph{S.insert()}
implements the one to one mapping function from $\signame{j}$ to $\pidname{j}$.  The
implementation details of \s\ will be discussed in Section~\ref{sec:more_discussion}.

For ease of discussion and investigation,
we assume in what follows that the node and edge tables are each just one file sequentially
filled with fixed length records.  Moreover, in this paper we make use of sort
merge join to the extent possible, since it is a very basic way to achieve
I/O efficient results.
However, many possibilities could be explored for implementing these data structures (e.g., indexing techniques) and join algorithms
to further optimize our presented results.
We leave such investigations open for future research.

Finally, we also assume that we have a (possibly external memory based) priority
queue available.  In our empirical study below, we use the off-the-shelf I/O efficient
priority queue implementation provided by the open source STXXL library~\cite{Dementiev2008}.

\subsection{Cost model}

Since our focus is on disk-resident datasets,
we use standard I/O complexity notions to analyze our algorithms~\cite{Aggarwal1988}.
The primary concern here is to minimize the number of I/Os needed to complete
the task at hand.

Suppose we have table $X$, space to hold $B$ disk pages in internal memory, and
$X$ occupies $|X|$ pages on disk.  In what follows, we will use the following
notation:

\begin{compactitem}

\item $sort(|X|)$ denotes the number of I/Os when sorting table $X$ on some given
    column(s). This will take $2|X|(1+\lceil
    log_{B-1}\lceil{\frac{|X|}{B}}\rceil\rceil)$ I/Os for a standard external memory based merge sort.

\item $scan(|X|)$ denotes the number of I/Os when scanning over table $X$. This will
take $|X|$ I/Os.
%\item $search(|X|)$ denotes the number of I/Os when searching for some key in some
%columns of $X$. This operation's cost will vary from a constant number of I/Os
%to $|X|$ I/Os, depending on how we implement $X$, whether $X$ is sorted on the
%certain columns, whether we have an index on those columns, and whether we want
%to search for the whole result set or the first key appearance.

\end{compactitem}

\section{Constructing localized \\\mbox{bisimulation} partitions}\label{sec:construction}
%\subsection{The construction algorithm}

\begin{algorithm*}[htbp]
\caption{Compute the \kbisim\ equivalence classes of a graph}
\label{algo:build_k-bisims}
\begin{algorithmic}[1]
\Procedure {Build\_Bisim}{\nodetable, \edgetable, $k$}
\If {$k=0$}\label{algo:build_k-bisims_k0start}
    \State fill in the \pidt{0}{nId} and \pidt{new}{nId} columns of \nodetable \Comment{$O(sort(|\nt|))+O(scan(|\nt|))}$\label{algo:build_k-bisims_pid0}
    \State \Return (\nodetable, \edgetable)
\EndIf \label{algo:build_k-bisims_k0end}
%\Else   \Comment{}
    \State (\nodetable,\edgetable)$\gets$\Call {Build\_Bisim}{\nodetable, \edgetable, $k-1$} \Comment{$k>0$, recursive call}\label{algo:build_k-bisims_preparestart}
    \If {$k=1$}
        \State \nodetable\ $\gets$ sort(\nodetable) by \emph{nId}                  \Comment{$O(sort(|\nt|))$}
        \State \edgetable\ $\gets$ sort(\edgetable) by \emph{tId}                  \Comment{$O(sort(|\et|))$}
    \EndIf
    \State scan \nodetable, move content of column \pidt{new}{nId} to \pidt{old}{nId} \Comment{$O(scan(|\nt|))$}
    \State fill in the \pidt{old}{tId} column of \edgetable \label{algo:build_k-bisims_etabletid}\Comment{$O(scan(|\et|))+O(scan(|\nt|))$}
    %\State construct signature $sig_k$ of node
    \State initialize \s
    \State \emph{F} $\gets \pi_{\alpha}(\edgetable),\textrm{~where~}\alpha=(sId,eLabel,\pidt{old}{tId})$
    \State \emph{F} $\gets $ sort($F$) by $sId,eLabel,\pidt{old}{tId}$, removing duplicates \Comment{$O(sort(|\et|))$}\label{algo:build_k-bisims_prepareend}
    \For {each $uId \in \pi_{nId}$(\nodetable)}\Comment{overall $O(scan(|\et|))+O(scan(|\nt|))+cost~of~S$}
    \label{algo:build_k-bisim_simple_forstart}
        \State construct \sig{k}{uId} from $F$ \label{algo:build_k-bisims_sigk}\Comment{merge join with $F$}
        \State \pid{k}{uId} $\gets$ \emph{S.insert}(\sig{k}{uId})
        \State record \pid{k}{uId} in \nodetable.\pidt{new}{nId} where $nId=uId$
    \EndFor\label{algo:build_k-bisim_simple_forend}
%\EndIf \label{algo:build_k-bisims_kg0end}
%\State fill in the $pId_k[nId]$ column of \nodetable
\State \Return (\nodetable, \edgetable)\label{algo:build_k-bisims_kg0end}
\EndProcedure
\end{algorithmic}
\end{algorithm*}

We present our algorithm for \kbisimnormal\ partition computation in Algorithm
\ref{algo:build_k-bisims}.  The algorithm is inspired by Proposition
\ref{prop:equal}, meaning for each node in the input graph, to construct its signature and find a
one-to-one mapping number (partition identifier) for that signature.

In iteration $j=0$, we assign distinct partition identifiers to nodes based on their \emph{nLabel}s.
For other iterations $j>0$, our algorithm mainly performs
two things for each node ID $uId \in \pi_{nId}(\nt)$ (line \ref{algo:build_k-bisim_simple_forstart} to
\ref{algo:build_k-bisim_simple_forend}): (1) construct \sig{j}{uId}; and (2)
insert \sig{j}{uId} to \s, record the returning \pid{j}{uId} in the corresponding
row in \nodetable.  To prepare the necessary information for constructing
\sig{j}{uId}, we need to fill in the missing columns of \edgetable\
(line \ref{algo:build_k-bisims_preparestart} to
\ref{algo:build_k-bisims_etabletid}).  Several scans and sorts on tables are
involved for each iteration.
Note that some operations in the algorithm can be merged as one in practice.
We present them separately just to make the presentation clearer.
A detailed description is given in Section~\ref{section:algorithm_description}.

%\newpage
\subsection{Details of Algorithm \ref{algo:build_k-bisims} (\construct)}\label{section:algorithm_description}
\subsubsection{Input and output}
The input variables of Algorithm \ref{algo:build_k-bisims} are node table \nodetable, edge table \edgetable\ and $k$, which is the degree of local bisimilarity from Definition \ref{def:kbisim}.
The output variables are \nodetable\ and \edgetable.
The schema of \nodetable\ is (\emph{nId}, \emph{nLabel}, \pidt{0}{nId}, \pidt{old}{nId}, \pidt{new}{nId});
the schema of \edgetable\ is (\emph{sId}, \emph{eLabel}, \emph{tId}, \pidt{old}{tId}).

%\begin{table}[htbp]
%\centering
%\caption{Tables for Algorithm \ref{algo:build_k-bisims}.
%The \nodetable\ is with schema (\emph{nId}, \emph{nLabel}, \pidt{0}{nId}, \pidt{old}{nId}, \pidt{new}{nId});
%the \edgetable\ is with schema (\emph{sId}, \emph{eLabel}, \emph{tId}, \pidt{old}{tId}).
%}
%%\singlespacing
%\footnotesize
%\tabcolsep=0.1cm
%\begin{subtable}{.45\textwidth}
%\centering
%\caption{\nodetable}
%\label{table:node_table1}
%\begin{tabular}{|l|l|l|l|l|}
%\hline
%\emph{nId} & \emph{nLabel} & \pidt{0}{nId} & \pidt{old}{nId} & \pidt{new}{nId} \\ \hline
%~ & ~ & ~ & ~ & ~\\
%\hline
%\end{tabular}
%\end{subtable}
%~
%\begin{subtable}{.45\textwidth}
%\centering
%\caption{\edgetable}
%\label{table:edge_table1}
%\begin{tabular}{|l|l|l|l|}
%\hline
%\emph{sId} & \emph{eLabel} & \emph{tId} & \pidt{old}{tId} \\ \hline
%~ & ~ & ~ & ~ \\
%\hline
%\end{tabular}
%\end{subtable}
%\end{table}
\newpage
\subsubsection{\boldmath{$k = 0$}, line \ref{algo:build_k-bisims_k0start} to \ref{algo:build_k-bisims_k0end}}

According to Definition \ref{def:kbisim}, $k = 0$ means nodes having the same labels should
be assigned the same partition identifier. We achieve this by sorting the
\nodetable\ on \emph{nLabel} column.
When scanning \nt, for each new \emph{nLabel} we encounter, we assign a new integer (e.g., a
predefined counter) to the corresponding \emph{nId}, filling it in the
\pidt{0}{nId} and \pidt{new}{nId} columns. This will take $O(sort(|\nt|)) + O(scan(|\nt|))$ I/Os.  Note
that since \pidt{0}{nId} can be assigned during the last step of the sorting
process, the scanning cost $O(scan(|\nt|))$ can be omitted.
Also note that one alternative way of assigning $pId_0$ is to use a hash map.
We can create a hash map using \emph{nLabel} as the keys, and $pId_0$ as the values.
The upper bound of this method is the same as the one we present here. %($O(sort(|\nt|))$).

\begin{leftbar}
\begin{algorithmic}%[1]
\State details of line \ref{algo:build_k-bisims_pid0} of Algorithm \ref{algo:build_k-bisims}:
\State sort \nodetable\ by \emph{nLabel} \Comment{$O(sort(|\nt|))$}
\State create variable \emph{current\_pId}
    \ForAll {$(nId,nLabel,\pidt{0}{nId},\pidt{old}{nId},\allowbreak \pidt{new}{nId})$\par
\hskip\algorithmicindent $\in \nodetable$} \Comment{$O(scan(|\nt|))$}
        \If {\emph{nLabel} is new}
            \State \emph{current\_pId} $\gets$ request a new \emph{pId}
        \EndIf
        \State save \emph{current\_pId} to \pidt{0}{nId} and \pidt{new}{nId}
    \EndFor
\end{algorithmic}
\end{leftbar}

\subsubsection{\boldmath{$k > 0$}, line \ref{algo:build_k-bisims_preparestart} to \ref{algo:build_k-bisims_kg0end}}

For $k>0$, we first perform a recursive call to the algorithm, ensuring we work in a bottom-up manner.
For iteration 1 ($k=1$), we sort \nodetable\ and \edgetable\ on \emph{nId} and \emph{tId}, preparing them for later merge join operations.
%In iteration $k$, we only need the information from iteration $k-1$ to do the \kbisim\ calculation.
The algorithm's idea is to construct the signature of each node in order to
distinguish it from other nodes according to the $k$-bisimilar relation.
If we can properly fill in the \pidt{old}{tId} column of \edgetable, and join it with \nodetable\ on \emph{nId=sId},
the information combined
from columns $\{\pidt{0}{nId}, eLabel, \pidt{old}{tId}\}$ is enough for constructing the signature.
The column \emph{eLabel} is already filled in before algorithm starts.
The column \pidt{0}{nId} is filled in during iteration 0 (line \ref{algo:build_k-bisims_k0start}
to \ref{algo:build_k-bisims_k0end}).
The column \pidt{old}{tId} is filled in during each iteration $j>0$ (line \ref{algo:build_k-bisims_etabletid}).
Then for each node ID $uId \in \nodetable$, we get its \sig{k}{uId}, insert it to \s\ in an I/O
efficient way, getting $pId_k(uId)$ in return,  and then placing this value in the
\pidt{new}{nId} column of \nodetable.

%At line \ref{algo:build_k-bisims_etablesid} of Algorithm \ref{algo:build_k-bisims}, to fill in the $pId_{0}[sId]$ column of \edgetable, we sort(\edgetable) by $sId$, do a sort merge join of \edgetable\ and \nodetable\, replacing the content of $pId_{0}[sId]$ in \edgetable\ with \pidt{0}{nId} in \nodetable\.
%
%\begin{leftbar}
%\begin{algorithmic}%[1]
%\State details of line \ref{algo:build_k-bisims_etablesid} of Algorithm \ref{algo:build_k-bisims}:
%    \State $edgeTable \gets $ sort(\edgetable) by $sId$; \Comment{$O(sort(|\et|))$}
%    \State $edgeTable \gets \pi_{\theta1} (edgeTable \Join_{\theta2} nodeTable)$,\Comment{$O(scan(|\et|))+O(scan(|\nt|))$}
%    \item[] $\theta1:$ use the schema of \edgetable, replace content of $pId_{0}[sId]$ in \edgetable\ with \pidt{0}{nId} in \nodetable\,
%    \item[] $\theta2:edgeTable.sId = nodeTable.nId$;
%\end{algorithmic}
%\end{leftbar}

At line \ref{algo:build_k-bisims_etabletid} of Algorithm \ref{algo:build_k-bisims}, to fill in the \pidt{old}{tId} column of
\edgetable, we conduct a sort merge join of \edgetable\ and \nodetable\ (since both tables are sorted properly in iteration 1), replacing the content of \pidt{old}{tId} in \edgetable\ with \pidt{old}{nId} in \nodetable.

\begin{leftbar}
\begin{algorithmic}%[1]
\State details of line \ref{algo:build_k-bisims_etabletid} of Algorithm \ref{algo:build_k-bisims}:
\State \edgetable $\gets \pi_{\alpha} (\edgetable \Join_{\phi} \nodetable)$
\Comment{merge join of \et\ and \nt} \label{algo:build_k-bisim_sort_merge}
%\item[] $\theta1:$ use the schema of \edgetable, replace content of $pId_{k-1}[tId]$ in \edgetable\ with $pId_{k-1}[nId]$ in \nodetable,
\item[] $\alpha:$(\edgetable.\emph{sId}, \edgetable.\emph{eLabel}, \edgetable.\emph{tId}, \nodetable.\pidt{old}{nId})
\item[] $\phi:$\edgetable.\emph{tId} = \nodetable.\emph{nId}
\end{algorithmic}
\end{leftbar}

At line \ref{algo:build_k-bisims_sigk} of Algorithm \ref{algo:build_k-bisims},
we sequentially construct the signature \sig{k}{uId} for each $uId \in \pi_{nId}(\nodetable)$ according to Definition \ref{def:sig}, and get the corresponding \pid{k}{uId} (using \emph{S.insert()}).
All \pid{k}{uId} will be written back to the \pidt{new}{nId} column of \nodetable\ (where \emph{nId}=\emph{uId}) right after, so that there is no random access to \nt.
Note that although by definition $sig_k$ is a set, we construct \sig{k}{uId} as a string, maintaining elements of the set in sorted order.
It is both an easy way for storing a set and handy for implementing \s\ later on (e.g., using a trie).

\begin{leftbar}
\begin{algorithmic}%[1]
\State details of line \ref{algo:build_k-bisims_sigk} of Algorithm \ref{algo:build_k-bisims}:
\State create string \sig{k}{uId} $\gets$ \pid{0}{uId} \Comment{overall scan \nodetable}
\If{$uId \in \pi_{sId}(F)$}
    \For {each (\emph{uId}, \emph{eLabel}, \pidt{old}{tId}) $\in$ \emph{F}}
\item[] \Comment{sequentially scan \emph{F}}
%\hskip\algorithmicindent where $u=sId$}
    \State \sig{k}{uId} $\gets$ \sig{k}{uId}+(\emph{eLabel}, \pidt{old}{tId})
    \EndFor %\Comment{\s\ is $sig_k(sId)$}
\EndIf

%\State Get $pId_k(sId)$ from $S.insert(s)$, and save it to $nodeTable.pId_k(nId)$ column;\Comment{$sId = nId$}\label{algo:build_k-bisim_insert}
\end{algorithmic}
\end{leftbar}

\subsection{Further discussion of Algorithm \ref{algo:build_k-bisims}}
\label{sec:more_discussion}
\paragraph*{Example run}

If we assume the numbering scheme for \s\ is a self-increased counter across
iterations, Table \ref{table:example_table} would be the intermediate results for
running Algorithm \ref{algo:build_k-bisims} on the example graph in Figure
\ref{fig:example_graph} ($k=2$), and Table \ref{table:example_output} gives the
final output of the algorithm.

\begin{table}[htbp]
\centering
%\tiny
\tabcolsep=0.1cm
\caption{Output of Algorithm \ref{algo:build_k-bisims} on example graph in Figure \ref{fig:example_graph} ($k=2$)
\label{table:example_output}
}
\vspace{-7pt}
\begin{subtable}[t]{.60\columnwidth}
\centering
\caption{\nodetable}
\label{table:node_table_output}
\resizebox{\textwidth}{!}{
\begin{tabular}{|c|c|c|c|c|}
\hline
\emph{nId} & \emph{nLabel} & \pidt{0}{nId} & \pidt{old}{nId} & \pidt{new}{nId} \\ \hline
1 & \emph{M} & 1 & 3 & 7\\
2 & \emph{M} & 1 & 3 & 8\\
3 & \emph{P} & 2 & 4 & 9\\
4 & \emph{P} & 2 & 5 & 10\\
5 & \emph{P} & 2 & 4 & 9\\
6 & \emph{P} & 2 & 6 & 11\\
\hline
\end{tabular}
}
\end{subtable}
~
\begin{subtable}[t]{.35\columnwidth}
\centering
\caption{\edgetable}
\label{table:edge_table_output}
\resizebox{\textwidth}{!}{
\begin{tabular}{|c|c|c|c|}
\hline
\emph{sId} & \emph{eLabel} & \emph{tId} & \pidt{old}{tId} \\ \hline
3 & \emph{l} & 1 & 3 \\
1 & \emph{w} & 2 & 3 \\
2 & \emph{w} & 2 & 3 \\
5 & \emph{l} & 2 & 3 \\
4 & \emph{l} & 3 & 4 \\
1 & \emph{l} & 4 & 5 \\
2 & \emph{l} & 6 & 6 \\
\hline
\end{tabular}
}
\end{subtable}
%\vspace{-5pt}
\end{table}

%about S and algorithm stopping
\paragraph*{Early stopping condition}
It is not always necessary to let the algorithm run $k$ iterations.
Indeed, it can be shown
\iftoggle{expand}{
(referring to Section \ref{sec:stopcondition} in Appendix)
}{
(proof omitted)
%do nothing in shorter version
}that after a bounded number of computation iterations, Algorithm \ref{algo:build_k-bisims} would achieve the full (i.e.,
classical non-localized) bisimulation partition.
We could detect this by simply checking the partition size each iteration produces. If two consecutive iterations produce the same number of partition
blocks, this means that the algorithm already achieves the full bisimulation partition, and therefore it is safe to terminate the algorithm.

\paragraph*{Numbering schemes of partition identifier and S}

In the algorithm, the correctness of the partition identifiers' assignment is
guaranteed level by level, meaning that the partition block numbering scheme from iteration $j$
has nothing to do with that of iteration $j+1$, for example. This means that we could use one
counter for the whole computation, or could use different counters for each
computation iteration.

The same idea also applies for implementing \s. As long as \s\ returns distinct
\emph{pId}s for different signatures for each computation iteration, it is
immaterial to the work performed by Algorithm \ref{algo:build_k-bisims}
if \s\ is a new one for each iteration or not. So,
we could use one \s\ for all iterations (when we have a global counter), to
reuse some signature \emph{pId} across iterations.  Furthermore, in practice
there could potentially be benefits from warm caching (get a better hit ratio)
for this approach.  Moreover, for the maintenance algorithms presented in Section
\ref{sec:maint}, we would only need to
store one \s\ instead of $k$ of them.  Essentially if the same signature appears
many times in different iterations, we only save it once in \s.  The drawback of
this method is that the size of \s\ will keep increasing as the algorithm runs.
This issue becomes acute when the number of partitions becomes large and the signatures
are long, as we observed in some datasets presented in Section \ref{sec:build_experiment}.
%In our experimental evaluations presented in Section \ref{section:experiment},

\paragraph*{Data structures for S}
The signature storage facility \s\ clearly plays an
important role in Algorithm \ref{algo:build_k-bisims}.  In principle, any data
structure that permits an efficient set-equality check will be sufficient.  Trie
and dictionary are such data structures, for instance.  During our experiments,
we see that in many of the cases, partition sizes are small and the signatures
are short, for which a main memory based data structure is enough.  In other cases,
signature length could reach several million and partition size
into tens of millions, then we need some external memory based
solution for \s.
We could, for example, sort all signatures from \emph{F} in an I/O efficient way~\cite{Arge1997},
then when scanning these signatures, partition identifiers are assigned.
In this case, the overall cost of the \emph{S.insert()} operation could still be bounded by $O(\mathit{sort}(|\et|))$.
Other disk based solutions, such as disk-based tries (e.g., String B-Tree \cite{FerraginaG99} or \cite{Roberto12}) or inverted files
(e.g., \cite{Mamoulis2003}) could also be considered.

In our experiments we use BerkeleyDB (B-Tree or Hash index) to mimic a trie, which,
as we show in the experimental results, has acceptable empirical behavior.

%\subsection{Complexity and correctness}
\paragraph*{Complexity and correctness}
We have the following characterization of Algorithm \ref{algo:build_k-bisims}.
\begin{theorem}
\label{theory:construct}
Let $k\geq 0$ and \graph\ be a graph.
Algorithm \ref{algo:build_k-bisims} computes the k-bisimulation partition
of $G$ with I/O complexity of $O(k\cdot \mathit{sort}(|\et|) + k\cdot \mathit{scan}(|\nt|) + \mathit{sort}(|\nt|))$,
and space complexity of $O(|\nt|+|\et|)$.
\end{theorem}
\iftoggle{expand}{
  A proof can be found in Appendix \ref{proof:construct}.
}{
Proof omitted.
%do nothing in shorter version
}

\paragraph*{Differences with Hellings et al}
As indicated in Section \ref{sec:intro}, the only known solutions for computing
bisimulation on graphs in external memory are those of
Hellings et al.\ \cite{Hellings:2012:EEB:2213836.2213899}.
There are two critical differences between their work and ours.
(1) {\em Targeting different problems.} The solutions of Hellings et al.\ are designed
specifically for the special case of acyclic graphs. Our approach does
not rely on such structure, computing bisimulation regardless of the presence or
absence of cycles in the graph.
(2) {\em Using different techniques.}  Hellings
et al.\ compute partition blocks level by level, starting from the leaf nodes of
the graph.   Our approach constructs all partition blocks at each iteration,
using data structures and processing strategies which are not tied to any (a)cyclic
structure in the graph.
In particular, the techniques of Hellings et al.\ do not generalize to
graphs having cyclic structure.

%\section{Maintenance on the localized bisimulation partition result}\label{sec:maint}
\section{Maintenance of localized \\\mbox{bisimulation} partitions}\label{sec:maint}
It is easy to show that any edge and node updates on a graph can potentially
change the complete \kbisimnormal\ partition of the graph. Therefore, in the worst case, the lower bound of such maintenance
cost is the cost of recomputing the \kbisimnormal\ partition from scratch.
However, when dealing with real graphs, as we shall see in Section \ref{section:experiment},
in many cases there is still hope to use data structures such as \s\ and priority queue to maintain the correct partition result instead of recomputing everything.
 In this section we propose several algorithms for this purpose.

For maintenance algorithms we assume that we have constructed the \kbisimnormal\
partition of graph \graph, where, as before,  $G$'s \nodetable\ and
\edgetable\ are stored on disk, containing the historical
information kept in \nodetable\ (Table \ref{table:node_table2});
\edgetable\ is the same as in Algorithm \ref{algo:build_k-bisims}, but
has two copies with sort orders (\emph{sId,tId}) and (\emph{tId,sId}) to boost performance.
We use \edgetable$_{st}$ and \edgetable$_{ts}$ to refer to each of these copies.

\begin{table}[htbp]
\centering
\caption{\nodetable\ for maintenance algorithms}
%, with schema $(nId,\allowbreak nLabel,\allowbreak\pidt{0}{nId},\allowbreak\ldots,\allowbreak \pidt{k}{nId})$.}
\label{table:node_table2}
%\singlespacing
\footnotesize
\tabcolsep=0.1cm
\begin{tabular}{|l|l|l|l|l|l|}
\hline
\emph{nId} & \emph{nLabel} & \pidt{0}{nId} & \pidt{1}{nId} & $\dots$ & \pidt{k}{nId} \\ \hline
~ & ~ & ~ & ~ & ~ & ~\\
\hline
\end{tabular}
\end{table}

We further assume that we save the signature storage facility \s\ on
disk, which we use and update throughout the maintenance process.

The maintenance problem includes the following subproblems.

\paragraph*{Change \boldmath{$k$}} If $k$ increases, we carry out another iteration of
computation. If $k$ decreases, the result can be returned directly since we keep the
history information in \nodetable.

%depending on how much history information we keep in tables, we could either answer it directly or recompute it from scratch (or from a certain step).

%\paragraph*{Add a new node \boldmath{$(uId, uLabel)$} (\addnode)}
%
%\begin{algorithm}[htbp]
%\caption {Add a new node to existing \kbisim\ partition}
%\label{algo:add_node}
%\begin{algorithmic}[1]
%\Procedure {Add\_Node}{\nodetable, \s, $(uId,uLabel), k$}
%\item[]\Comment{$uId$ is the new node ID, $uLabel$ is its node label}
%\State search for row $(vId,vLabel,\dots)\in \nodetable$ where \par
%\hskip\algorithmicindent$uLabel=vLabel$ \Comment{$O(search(|\nt|))$}
%\If {could find row $(vId,vLabel,\dots)$}
%    \State use \pidt{0}{nId} of $vId$ for \pid{0}{uId}
%\Else
%    \State request a new $pId$, use it for \pid{0}{uId}
%\EndIf
%\State get value of \emph{S.insert}(\pid{0}{uId}), use it for \par
%\hskip\algorithmicindent$pId_1,\ldots,pId_k$ of $uId$ \Comment{some constant I/O}
%\State insert $(uId,uLabel,\pid{0}{uId},\dots,\pid{k}{uId})$ to \nodetable
%\State \Return (\nodetable, \s)
%\EndProcedure
%\end{algorithmic}
%\end{algorithm}

\begin{algorithm*}[htbp]
\caption {Add a set of new nodes to existing \kbisim\ partition}
\label{algo:add_nodes}
\begin{algorithmic}[1]
\Procedure {Add\_Nodes}{\nodetable, \s, \emph{newNodes}, $k$}
\Comment{\emph{newNodes} is a table of new nodes}
\State \nodetable $\gets$ sort(\nodetable) by $nLabel$ \Comment{$O(\mathit{sort}(|\nt|))$}
\State \emph{newNodes} $\gets$ sort(\emph{newNodes}) by $nLabel$ \Comment{$O(\mathit{sort}(|\nt|))$}
\State \emph{newNodes} $\gets \pi_{\alpha}(\mathit{newNodes} \leftouterjoin_{\phi} (\nodetable))$, remove duplicates \Comment{$O(\mathit{scan}(|\nt|))$}
\item[] $\alpha:$(\emph{newNodes.nId}, \emph{newNodes.nLabel},%\par
\nodetable.\pidt{0}{nId}, \dots)
\item[] $\phi:newNodes.nLabel = \nodetable.nLabel$
%\item[] $\beta:(nLabel, \pidt{0}{nId})$
\State request a new \emph{pId} for each new \emph{nLabel} in \emph{newNodes}, %\par
fill in all the NULL fields in \emph{newNodes}.\pidt{0}{nId}
\For {each $uId \in \pi_{nId}(\mathit{newNodes})$} \Comment{overall $O(\mathit{scan}(|\nt|))$ + cost of \s}
    \State get value of \emph{S.insert}(\pid{0}{uId}), use it for %\par
    $\pidt{1}{nId},\allowbreak \dots,\pidt{k}{nId}$ of \emph{uId}
\EndFor
\State append \emph{newNodes} to \nodetable
\State \Return (\nodetable, \s)
\EndProcedure
\end{algorithmic}
\end{algorithm*}

%When adding a new node $(uId,uLabel)$ ($uId$ is the new node's ID, $uLabel$ is its node label, resp.), we assume the node is isolated.
%In this case, we will not modify \edgetable, but will insert one
%row $(uId,uLabel,\pid{0}{uId},\allowbreak ...,\allowbreak \pid{k}{uId})$ to \nodetable.
%For \pid{0}{uId}, we search for some row $(vId,vLabel,\pid{0}{vId},\allowbreak ...,\allowbreak \pid{k}{vId}) \in \nodetable$
% such that $vLabel = uLabel$, then we
%assign \pid{0}{vId} to \pid{0}{uId}. If we cannot find such $vId$, we request a
%new $pId$ and use it for \pid{0}{uId}.  For $k>0$, since \sig{j}{uId}
%$(j\in\{1,\ldots,k\})$ is always $(\pid{0}{uId}, \emptyset)$, we use the value of
%$S.insert(pId_0(uId))$ for \pid{j}{u}.  Since the cost of inserting
%one such signature to \s\ requires a constant number of I/Os,
%the algorithm's cost is bounded by $O(search(|\nt|))$.  Pseudo code is in Algorithm
%\ref{algo:add_node}.

\paragraph*{Add a set of new nodes (\addnodes)}
When adding a set of new nodes, we assume the new nodes are isolated, stored in the \emph{newNodes} table, which
has the same schema as \nodetable, and
that $|\mathit{newNodes}|=O(|\nodetable|)$.
We first sort \nodetable\ and \emph{newNodes} by
\emph{nLabel}, then perform a merge join on the \emph{nLabel} column to fill in the
\pidt{0}{nId} column of \emph{newNodes} for all the existing \emph{nLabel}. For
the missing ones, we request a new \emph{pId} for each of the new \emph{nLabel}.
Then we get the $pId_1,\ldots,pId_k$ of the \emph{newNodes} by inserting its
$pId_0$ to \s.  At the end we append the whole \emph{newNodes} to \nodetable.
The I/O complexity of \addnodes\ is bounded by $O(sort(|\nt|))$.  Pseudo code is in Algorithm
\ref{algo:add_nodes}.

\paragraph*{Add a set of new edges (\addedges)}

For adding a set of edges, we assume that the edges are added between existing nodes.
If this is not the case, we first call procedure \addnodes.
The new edges are stored in the \emph{newEdges} table, having the same schema as \et.
For inserting one edge $(s,l,t)$ to $G$, the potential changes are to \sig{j}{s} $(1\le j \le k)$, as well
as those signatures of all ancestors of $s$ within $k$ steps.
So the main work is to detect whether there is some change in \sig{j}{s} and propagate
those change(s) to its parent nodes' signatures in later iterations.  We use a
priority queue \textbf{\pq}\ to record and process such changes in a systematic,
level-wise manner.
For some node ID \emph{uId} and
iteration \emph{j}, \pq\ stores the pair \emph{(j,uId)} as priority reference.
Then whenever we dequeue one element from
\pq, we get the smallest node ID from the lowest iteration (lowest priority
reference). Therefore
\pq\ indicates those nodes whose signatures could change in each iteration level
(from 1 up to $k$).

At the beginning of the algorithm, we enqueue $(j,s)$ to \pq\ $(\forall (s,l,t)\in \et, 0 < j \le k)$.
Then, while \pq\ is not empty, we dequeue the list of $(j,uId)$ pairs with the
same $j$ out of the queue, construct the new signature of each such \emph{uId}, insert it to
\emph{S}, and compare the returning \pid{j}{uId} with the old \pidt{j}{nId} value of \emph{uId}.  If the \emph{pId}
remains the same as the old one, we continue; if it changes, we record
\pid{j}{uId} in \nodetable, and enqueue all $(j+1,\mathit{vId})$ pairs to \pq\ where \emph{vId}~$\in
\pi_{sId}(\sigma_{tId=uId}(\edgetable))$. Pseudo code is given in Algorithm
\ref{algo:add_edge}, and a detailed discussion is in Section
\ref{section:edge_algo_description}.

\paragraph*{Deletions}

Deletions follow a similar idea to insertions.
For example, when removing an edge $(s,l,t)$, it is the same idea as adding one.
We also (potentially) modify the
signature of $s$, propagating changes to its ancestors via \pq, then the reasoning is the same.
When removing a node, we first remove each incoming edge and each outgoing edge for that node. Then we
remove the node from the node table.

\begin{algorithm*}[thbp]
\caption{Add a set of new edges to existing \kbisim\ partition}
\label{algo:add_edge}
\begin{algorithmic}[1]
\Procedure {Add\_Edges}{\nodetable, \edgetable$_{st}$, \edgetable$_{ts}$, \s, $newEdges, k$}
\Comment{$newEdges$ is a table of new edges}
\If {$k=0$}\label{algo:add_edge_k0start}
  \State merge \emph{newEdges} into \edgetable$_{st}$ and \edgetable$_{ts}$\Comment{$O(\mathit{sort}(|\et|))$}\label{algo:add_edge_k0end}
  %\State insert ($s,l,t$) to \edgetable$_{st}$ \label{algo:add_edge_k0end}\Comment{$O(search(|\et|))$}
\Else                                  \Comment{$k>0$} \label{algo:add_edge_kge1start}
  \State \nodetable $\gets $ sort(\nodetable) by \emph{nId}                  \Comment{$O(\mathit{sort}(|\nt|))$}
  \State create empty priority queue \pq \Comment{overall $O(\mathit{sort}(|\nt|))$}
                %\item[]\Comment{\pq\ takes $(j,nId)$ as the input element and priority reference, where $j$ is the iteration number, $nId$ is the node id}
  \For {$j \in \{1,\ldots,k\}$ and $(s,l,t) \in$ \emph{newEdges}}\label{algo:add_edge_enqueue_start}
    \State enqueue $(j,s)$ to \pq \label{algo:add_edge_oldsig}
  \EndFor\label{algo:add_edge_enqueue_end}
  %\State create variable $j\gets 0$
  %\State insert ($s,l,t,pId_0(t)$) to \edgetable$_{st}$ and \edgetable$_{ts}$ \Comment{$O(search(|\et|))$}
  \State merge \emph{newEdges} into \edgetable$_{st}$ and \edgetable$_{ts}$, fill in the \pidt{old}{tId} column\Comment{$O(\mathit{sort}(|\et|))$}
  \While {\pq\ is not empty}
    \State dequeue all pairs $(j,uId)$ from \pq\ with the same (i.e., smallest) $j$ value, save all distinct \emph{uId} to \emph{M}%\item[]\Comment{When dequeue, remove duplicates}
    \State \emph{F} $\gets \sigma_{sId \in \mathit{M}}$(\edgetable$_{st})$ \Comment{merge join, $O(\mathit{scan}(|\nt|)+\mathit{scan}(|\et|))$}
    \State fill in the \pidt{old}{tId} column of $F$ \label{algo:add_edge_filltid}\Comment{$O(\mathit{scan}(|\nt|)+O(\mathit{sort}(|\et|))+O(\mathit{scan}(|\et|)))$}
    \State $\mathit{H} \gets \pi_{\alpha}(\mathit{F})$, where $\alpha$=(\emph{sId}, \emph{eLabel}, \pidt{old}{tId})
    \State \emph{H} $\gets$ sort \emph{H} on \emph{sId}, \emph{eLabel}, \pidt{old}{tId}, and remove duplicates\Comment{$O(\mathit{sort}(|\et|))$}
    %\item[]
    \ForAll {\emph{uId} $\in$ \emph{M}}\Comment{scan \emph{M}, \nodetable\ and \emph{H}, overall $O(\mathit{scan}(|\nt|))+O(\mathit{scan}(|\et|))$ + cost of \s}
        \State construct \sig{j}{uId} from \emph{H} \label{algo:add_edge_sig}
        \State \pid{j}{uId} $\gets$ \emph{S.insert}(\sig{j}{uId})
        \If {\pid{j}{uId} is not the same as the corresponding value in \nodetable.\pidt{j}{nId}}
            \State propagate changes to \nodetable\ and \pq \label{algo:add_edge_propagate}\Comment{$O(\mathit{scan}(|\nt|))+O(\mathit{scan}(|\et|))$}
        \EndIf
    \EndFor
  \EndWhile\label{algo:add_edge_kge1end}
\EndIf
\State \Return (\nodetable, \edgetable$_{st}$, \edgetable$_{ts}$, \s)
\EndProcedure
\end{algorithmic}
\end{algorithm*}

\newpage

\subsection{Details of Algorithm \ref{algo:add_edge} (\addedges)}\label{section:edge_algo_description}

\subsubsection{Input and output}

The input variables of Algorithm \ref{algo:add_edge} are node table \nodetable,
edge tables \edgetable$_{st}$ and \edgetable$_{ts}$, the signature storage facility \s, the new edge set \emph{newEdges} and $k$.
The output variables of Algorithm \ref{algo:add_edge} are \nodetable,
\edgetable$_{st}$, \edgetable$_{ts}$ and \s.  \nodetable's schema is given in Table
\ref{table:node_table2}, while \edgetable$_{st}$, \edgetable$_{ts}$ and \emph{newEdges}'s schema is the same as \edgetable\ in Algorithm
\ref{algo:build_k-bisims}.

%\begin{itemize*}
%  \item \nodetable, with schema $(nId,nLabel, \pidt{0}{nId}, \allowbreak \ldots, \allowbreak \pidt{k}{nId})$, which means to keep all history information in the table;
%  \item \edgetable$_{st}$ and \edgetable$_{ts}$, with schema $(sId,eLabel,\allowbreak tId,\allowbreak \pidt{old}{tId})$, leaving \pidt{old}{tId} column blank initially, sorted on (\emph{sId,tId}) and (\emph{tId,sId});
%  \item \s, signature storage facility \s;
%  \item $(s,l,t)$, the new edge, where $s$ is the source node, $l$ is the edge label, and $t$ is the target node;
%  \item $k$, the number of computation iterations. To keep it simple, here we assume that this $k$ does not exceed the original computation iterations we perform in Algorithm \ref{algo:build_k-bisims}.
%\end{itemize*}
%The output variables of Algorithm \ref{algo:add_edge} are:
%\begin{itemize*}
%  \item \nodetable, with the same schema, with all information updated accordingly;
%  \item \edgetable$_{st}$, with the same schema, where information from columns (\emph{sId,eLabel,tId}) is updated accordingly, sorted on (\emph{sId,tId});
%  \item \s.
%\end{itemize*}

\subsubsection{\boldmath{$k=0$}, line \ref{algo:add_edge_k0start} to
\ref{algo:add_edge_k0end} of Algorithm \ref{algo:add_edge}}

For $k=0$, since all nodes' information is properly filled (including the
\pidt{0}{nId} column) in \nodetable\, we only need to add new rows to \edgetable$_{st}$ and \edgetable$_{ts}$ according to \emph{newEdges}.

\subsubsection{\boldmath{$k>0$}, line \ref{algo:add_edge_kge1start} to
\ref{algo:add_edge_kge1end} of Algorithm \ref{algo:add_edge}}

For $k>0$, for each iteration, which is indicated by $j$ in the algorithm,
we need to
(1) find out the potential nodes whose signatures could have changed;
(2) check whether these signatures have been changed or not;
and, (3) propagate any such changes to the parents of these nodes.
To record the potential nodes and to perform the propagation, we use a priority queue \pq.
\pq\ takes $(j,\mathit{nId})$ as the element and priority reference, where $j$ is the
iteration level and \emph{nId} is the node identifier.
To check signature changes, we reuse the signature storage facility \s.

When adding a new edge $(s,l,t) \in newEdges$ to the graph, all \sig{j}{s} $(j>0)$ have
the potential to change, and hence we add all pairs $(j,s)$, for $j \in \{1,\ldots,k\}$, to \pq,
indicating that we need to check the signature of $s$ in every iteration (line
\ref{algo:add_edge_enqueue_start} to \ref{algo:add_edge_enqueue_end}).
For each iteration $j>0$, we dequeue from \pq\ all node IDs in the smallest iteration $j$, remove duplicates, and save them to a
temporary table \emph{M}, so that \emph{M} contains in sorted order all node IDs whose signatures would change in iteration $j$.
Then we create an extra table \emph{F}, preparing for signature constructions.
This is achieved by performing a merge join of \edgetable$_{st}$\ and \emph{M} (where \edgetable$_{st}.sId
\in \mathit{M}$).
Then we fill in $\mathit{F}.\pidt{old}{tId}$ column, as in Algorithm \ref{algo:build_k-bisims}.

\begin{leftbar}
\begin{algorithmic}
\State details of line \ref{algo:add_edge_filltid} of Algorithm \ref{algo:add_edge}:
\State $\mathit{F} \gets \mathit{sort(F)}$ by $tId$  \Comment{$O(\mathit{sort}(|\et|))$}
\State $\mathit{F} \gets \pi_{\alpha}(\mathit{F} \Join_{\phi} \nodetable)$  \Comment{$O(\mathit{scan}(|\et| + |\nt|))$}
\item[] $\alpha:$ (\emph{F.sId}, \emph{F.eLabel}, \emph{F.tId}, \nodetable.\pidt{(j-1)}{nId})
\item[] $\phi:$ \emph{F.tId} = \nodetable.\emph{nId}
\end{algorithmic}
\end{leftbar}

After projection on the (\emph{sId}, \emph{eLabel}, \pidt{old}{tId}) of \emph{F} and removing
duplicates, we get \emph{H}, and are ready to construct the signatures.  For each
\emph{uId} $\in$ \emph{M}, we construct \sig{j}{uId} according to the signature definition.
The idea of constructing the nodes' signatures is the same as line \ref{algo:build_k-bisims_sigk} of Algorithm \ref{algo:build_k-bisims},
only in this case we are not considering every node
but only those appearing in \pq\ (and later in \emph{M}).

We then call \s.\emph{insert(\sig{j}{uId})} for all such $\mathit{uId}$. If \s\ returns the same \pid{j}{uId} as
recorded in \nt.\pidt{j}{nId}, nothing will happen; otherwise we change the \nt.\pidt{j}{nId} entry of \emph{uId} accordingly, and propagate the changes to \pq.
If $j < k$, we add
all parents of \emph{uId} to \pq\ to indicate that we will check these nodes' signatures
in the $j+1$ iteration. This is achieved by a merge join with \edgetable$_{ts}$.

\begin{leftbar}
\begin{algorithmic}
\State details of line \ref{algo:add_edge_propagate} of Algorithm \ref{algo:add_edge}:
\State record the new \pid{j}{uId} in the corresponding row in \nodetable
\item[] \Comment{overall $O(\mathit{scan}(|\nt|))$}
\If {$j < k$}
    \State $I \gets \sigma_{tId=uId}$(\edgetable$_{ts})$ \Comment{overall $O(\mathit{scan}(|\et|))$}
    \ForAll {(\emph{sId}, \emph{eLabel}, \emph{tId}, \pidt{old}{tId}) $\in$ \emph{I}} \item[] \Comment{overall $O(\mathit{scan}(|\et|))$}
        \State enqueue $(j+1,\mathit{sId})$ to \pq
    \EndFor
\EndIf
\end{algorithmic}
\end{leftbar}

%\subsection{Complexity and correctness of Algorithm \ref{algo:add_edge}}

\paragraph*{Complexity and correctness}
We have the following characterization of Algorithm \ref{algo:add_edge}.
\begin{theorem}
\label{theory:update}
Let \graph\ be a graph and $k\geq 0$.  After adding a set of new edges to $G$, Algorithm
\ref{algo:add_edge} correctly updates the \kbisimnormal\ partition of $G$ with I/O
complexity of $O(k\cdot \mathit{sort}(|\et|) + k\cdot \mathit{sort}(|\nt|))$,
and space complexity of $O(k\cdot|\nt|+k\cdot|\et|)$.
\end{theorem}
\iftoggle{expand}{
  A proof can be found in Appendix \ref{proof:update}.
}{
Proof omitted.
%do nothing in shorter version
}

\subsection{Further discussion of Algorithm \ref{algo:add_edge}}
\paragraph*{Example run}
We present different behaviors of Algorithm \ref{algo:add_edge} using two
examples.  Here we will extend the graph from Figure \ref{fig:example_graph}
as in Figure \ref{fig:new_example_update}.  The
dashed lines in this figure indicate the two edges which we will add in
our examples.

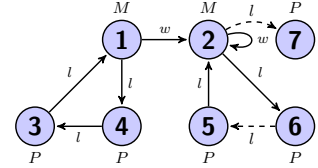
\begin{figure}[htbp]
\centering
\resizebox{.5\columnwidth}{!}{
\begin{tikzpicture}[->,>=stealth',shorten >=1pt,auto,node distance=1.7cm,
  thick,main node/.style={circle,fill=blue!20,draw,font=\sffamily\Large\bfseries}]
\centering
\useasboundingbox (-2,-2) rectangle (4,1);
  \node[main node] (1) [label=above:$M$] {1};
  \node[main node] (2) [label=above:$M$,right of=1] {2};
  \node[main node] (4) [label=below:$P$,below of=1] {4};
  \node[main node] (3) [label=below:$P$,left of=4] {3};
  \node[main node] (5) [label=below:$P$,below of=2] {5};
  \node[main node] (6) [label=below:$P$,right of=5] {6};
  \node[main node] (7) [label=above:$P$,right of=2] {7};

  \path[every node/.style={font=\sffamily\small}]
    (1) edge node {$w$} (2)
    (1) edge node {$l$} (4)
    (2) edge [loop right] node {$w$} (2)
    (2) edge node {$l$} (6)
    (3) edge node {$l$} (1)
    (4) edge node {$l$} (3)
    (5) edge node {$l$} (2);
  \draw[dashed] (2) edge [bend left] node {$l$} (7);
  \draw[dashed] (6) edge node {$l$} (5);
%\draw (current bounding box.south west) rectangle (current bounding box.north east);

\end{tikzpicture}
}
%\vspace{-10pt}
\caption{Updates on the example graph}
\label{fig:new_example_update}
\end{figure}

First suppose we add edge $(2,l,7)$ to the original graph of Figure
\ref{fig:example_graph}, where node 7 is a new node with label \emph{P}.  Table
\ref{table:example_update1_table} shows the resulting partition after this
insertion.  The new/changed part of the table is indicated in gray.  When the
algorithm starts, (2,1) and (2,2) are added to \pq. Then after checking each of
these, the algorithm finds no change in node 2's signature, therefore no change
propagates, and the algorithm stops.  We see that comparing with Table
\ref{table:example_table}, the only thing that changes is to add one more row
(node 7) to the table.
Since node 7 does not have outgoing edges,
adding one edge that points into node 7 will not change any existing nodes's signature.
Node 7 belongs to the group of node 6, and no other node changes group membership.
%Here node 7 does not have any outgoing edge; it does not
%change any existing node's signature, and belongs to the same partition block as
%node 6.

In the second case, suppose we add edge $(6,l,5)$ to the original graph of
Figure \ref{fig:example_graph}.  The algorithm first add (6,1) and (6,2) to \pq.
Then in iteration 1, the algorithm detects that the signature of node 6 does
change, and therefore adds one new pair (2,2) to \pq. In iteration 2, both node 2 and
node 6's signatures are checked, and they are both changed.  We see that in
Table \ref{table:example_update2_table} \pid{2}{1} and \pid{2}{2} become the
same, while \pid{2}{6} changes from 11 to 10.

\begin{table}[thbp]
\centering
\caption{$2$-\emph{bisimulation} for the example graph after edge insertion $(2,l,7)$}
\label{table:example_update1_table}
%\singlespacing
\footnotesize
\tabcolsep=0.1cm
\resizebox{\columnwidth}{!}{
\begin{tabular}{|c|c|c|c|c|c|}
\hline
\emph{nId} & \pid{0}{nId} & \sig{1}{nId} & \pid{1}{nId} & \sig{2}{nId} & \pid{2}{nId}\\
\hline
1 & 1 & $1,\{(w,1),(l,2)\}$   & 3 & $1,\{(w,3),(l,5)\}$ & 7\\
2 & 1 & $1,\{(w,1),(l,2)\}$   & 3 & $1,\{(w,3),(l,6)\}$ & 8\\
3 & 2 & $2,\{(l,1)\}$        & 4 & $2,\{(l,3)\}$      & 9\\
4 & 2 & $2,\{(l,2)\}$        & 5 & $2,\{(l,4)\}$      & 10\\
5 & 2 & $2,\{(l,1)\}$        & 4 & $2,\{(l,3)\}$      & 9\\
6 & 2 & $2,\{\}$            & 6 & $2,\{\}$          & 11\\
\rowcolor{lightgray} 7 & 2 & $2,\{\}$            & 6 & $2,\{\}$          & 11\\
\hline
\end{tabular}
}
\end{table}

\begin{table}[thbp]
\centering
\caption{$2$-\emph{bisimulation} for the example graph after edge insertion $(6,l,5)$}
\label{table:example_update2_table}
%\singlespacing
\footnotesize
\tabcolsep=0.1cm
\resizebox{\columnwidth}{!}{
\begin{tabular}{|c|c|c|c|c|c|}
\hline
\emph{nId} & \pid{0}{nId} & \sig{1}{nId} & \pid{1}{nId} & \sig{2}{nId} & \pid{2}{nId}\\
\hline
1 & 1 & $1,\{(w,1),(l,2)\}$   & 3 & $1,\{(w,3),(l,5)\}$ & 7\\
2 & 1 & $1,\{(w,1),(l,2)\}$   & 3 & \cellcolor{lightgray}$1,\{(w,3),(l,5)\}$ & \cellcolor{lightgray}7\\
3 & 2 & $2,\{(l,1)\}$        & 4 & $2,\{(l,3)\}$      & 9\\
4 & 2 & $2,\{(l,2)\}$        & 5 & $2,\{(l,4)\}$      & 10\\
5 & 2 & $2,\{(l,1)\}$        & 4 & $2,\{(l,3)\}$      & 9\\
6 & 2 & \cellcolor{lightgray}$2,\{(l,2)\}$        & \cellcolor{lightgray}5 & \cellcolor{lightgray}$2,\{(l,4)\}$      & \cellcolor{lightgray}10\\
\hline
\end{tabular}
}
\end{table}

\paragraph*{When to switch back to Algorithm~\ref{algo:build_k-bisims}}
As we will see in our empirical study (Section \ref{sec:add_edges}), it is not always beneficial to use Algorithm~\ref{algo:add_edge}, since it performs extra work in each iteration.
 Heuristics could be adopted to decide when to switch back to Algorithm~\ref{algo:build_k-bisims}.
For example, if at a certain iteration, most of the nodes are placed into \pq, it is more beneficial to switch back to Algorithm~\ref{algo:build_k-bisims}.
This could be done by simply checking the size of \pq\ at the beginning of each iteration.

\vspace{-0.5pt}

\section{Empirical Analysis}\label{section:experiment}

In this section we present the results of an in-depth experimental study of our
algorithms.  After introducing our set-up, we first discuss a validation of the
correctness of our algorithms by several experiments. We then show the
performance of the algorithms on both synthetic and real datasets. In these
experiments, various
aspects of the algorithms are investigated while other settings are fixed.

\subsection{Experiment setting}

\paragraph*{Environment}
The following experiments are run on
a machine with 2.27~GHz Intel Xeon (L5520, 8192KB cache) processor, 12GB main
memory, running Fedora 14 (64-bit) Linux.  We use C++ to implement all the algorithms,
using GCC 4.4.4 as the compiler. We use the open-source STXXL library \cite{Dementiev2008} to
construct the tables and perform the external memory sorting, and use Berkeley DB to
implement \s.
One \s\ is used for all computation iterations (as discussed in Section \ref{sec:more_discussion}).
In the experiments we do not exploit any parallelism and restrain ourselves with
predefined buffer sizes.  We record the running time as well as the I/O volume
between the buffer and the disk system.  Therefore, the performance (time) of
the experiments are comparable to a commodity PC, and the I/O volume can
be repeated on other systems.  In the following experiments, we set both the
STXXL buffer and Berkeley DB buffer to be 128MB, if not otherwise indicated.
Please note that we run experiments for the Twitter dataset on a different
machine (Intel Xeon E5520, 2.27 GHz, 8192KB cache, 70G main memory, same OS) for limited
disk space reason,  using a 512MB/512MB buffer setting.

\paragraph*{Datasets}

To prove the practicability of the algorithms, we experiment with various graph
datasets. The datasets are collected from public repositories, ranging from synthetic
data to real-world data, from several million of edges to more than 1.4 billion
edges. In Table \ref{table:datasets} we give a description of the datasets, as
well as some simple statistics of them. All datasets are accessed on 15 May
2012.
Note that due to space limitation, in the following we show the experiment
results on a subset of the datasets when the result is representative enough.

\begin{table}[htbp]
\centering
\caption{Description and statistics of the experiment datasets}
\label{table:datasets}
%\singlespacing
\scriptsize
\tabcolsep=0.1cm
\resizebox{\columnwidth}{!}{
%\begin{tabular}{@{}r@{}|p{2.7cm}|@{}r@{}|@{}r@{}|@{}r@{}}
\begin{tabular}{|r|p{2.7cm}|r|r|r|}
  \hline
  Data Name & Description & Node Count & Edge Count & Label on \\
  \hline \hline
  Jamendo & A repository of music metadata in RDF format\footnotemark[1]  & 486,320 & 1,049,647 & Edge \\
  \hline
  LinkedMDB & A repository of movie metadata in RDF format~\cite{HassanzadehC09} & 2,330,695 & 6,147,996 & Edge \\
  \hline
  DBLP & An RDF format DBLP dump\footnotemark[2] & 23,000,670 & 50,203,406 & Edge \\
  \hline
  WikiLinks & A page-to-page linking graph of Wikipedia\footnotemark[3]  & 5,710,993 & 130,160,392 & None \\
  \hline
  DBPedia & An early RDF dump of DBPedia\footnotemark[4] & 38,615,135 & 115,305,444 & Edge \\
  \hline
  Twitter & A following relationship graph of Twitter~\cite{Kwak10www} & 41,652,230 & 1,468,365,182 & None \\
  \hline
  SP2B & A RDF data generator for arbitrarily large DBLP-like data~\cite{Schmidt2009}  & 280,908,393 & 500,000,912 & Edge \\
  \hline
  BSBM & A RDF data generator for e-commerce use case~\cite{bizer2009berlin} & 8,886,078 & 34,872,182 & Edge \\
  \hline
\end{tabular}
}
\end{table}

\footnotetext[1]{\url{http://dbtune.org/jamendo/}}
%\footnotetext[2]{\url{http://www.linkedmdb.org/}}
\footnotetext[2]{\url{http://thedatahub.org/dataset/l3s-dblp}}
\footnotetext[3]{\url{http://haselgrove.id.au/wikipedia.htm}}
\footnotetext[4]{\url{http://www.cs.vu.nl/~pmika/swc/btc.html}}
%\footnotetext[6]{\url{http://an.kaist.ac.kr/traces/WWW2010.html}}
%\footnotetext[7]{\url{http://dbis.informatik.uni-freiburg.de/index.php?project=SP2B}}
%\footnotetext[8]{\url{http://www4.wiwiss.fu-berlin.de/bizer/BerlinSPARQLBenchmark/}}
\footnotetext[5]{\url{http://snap.stanford.edu/data/index.html}}

\begin{figure*}[htbp]
\centering
\begin{subfigure}[t]{0.3\textwidth}
\raggedleft
\begin{tikzpicture}[baseline]%[baseline,trim axis left,trim axis right]
	\begin{semilogyaxis}[
tiny,
xlabel=Iteration,
ylabel=Partition count,
grid=major,
xtick={1,2,3,4,5,6,7,8,9,10},
height=5cm,width=5cm,
legend columns=8,
legend entries={Jamendo,LinkedMDB,DBLP,WikiLinks,DBPedia,Twitter,BSBM,SP2B},
legend to name=bignamed,
]
	\addplot table[x=Iteration,y=jamendo] {figures/partition_count.dat};
    \addplot table[x=Iteration,y=linkedmdb] {figures/partition_count.dat};
    \addplot table[x=Iteration,y=dblp] {figures/partition_count.dat};
    \addplot table[x=Iteration,y=wikilinks] {figures/partition_count.dat};
    \addplot table[x=Iteration,y=dbpedia] {figures/partition_count.dat};
    \addplot table[x=Iteration,y=twitter] {figures/partition_count.dat};
    \addplot table[x=Iteration,y=bsbm] {figures/partition_count.dat};
    \addplot table[x=Iteration,y=sp2b] {figures/partition_count.dat};
    %customize style
    %\addplot[mark=triangle*,gray,dashed] table[x=Iteration,y=bsbm] {figures/partition_count.dat};
    %\addplot[mark=pentagon*] table[x=Iteration,y=sp2b] {figures/partition_count.dat};
    \end{semilogyaxis}
\end{tikzpicture}
\caption{Number of partition blocks for every iteration}
\label{fig:partition_count}
\end{subfigure}
~
\begin{subfigure}[t]{0.3\textwidth}
\raggedleft
\begin{tikzpicture}[baseline]%[baseline,trim axis left,trim axis right]
	\begin{semilogyaxis}[
tiny,
xlabel=Iteration,
ylabel=Max signature length (integer count),
grid=major,
xtick={1,2,3,4,5,6,7,8,9,10},
height=5cm,width=5cm,
]
	\addplot table[x=Iteration,y=jamendo] {figures/max_signature.dat};
    \addplot table[x=Iteration,y=linkedmdb] {figures/max_signature.dat};
    \addplot table[x=Iteration,y=dblp] {figures/max_signature.dat};
    \addplot table[x=Iteration,y=wikilinks] {figures/max_signature.dat};
    \addplot table[x=Iteration,y=dbpedia] {figures/max_signature.dat};
    \addplot table[x=Iteration,y=twitter] {figures/max_signature.dat};
    \addplot table[x=Iteration,y=bsbm] {figures/max_signature.dat};
    \addplot table[x=Iteration,y=sp2b] {figures/max_signature.dat};
	\end{semilogyaxis}
\end{tikzpicture}
\caption{Maximum length of signatures for every iteration}
\label{fig:max_signature}
\end{subfigure}
~
\begin{subfigure}[t]{0.3\textwidth}
\raggedleft
\begin{tikzpicture}[baseline]%[baseline,trim axis left,trim axis right]
	\begin{semilogyaxis}[
tiny,
xlabel=Iteration,
ylabel=I/O on STXXL (byte),
grid=major,
xtick={1,2,3,4,5,6,7,8,9,10},
height=5cm,width=5cm,
]
	\addplot table[x=Iteration,y=jamendo] {figures/stxxl_io.dat};
    \addplot table[x=Iteration,y=linkedmdb] {figures/stxxl_io.dat};
    \addplot table[x=Iteration,y=dblp] {figures/stxxl_io.dat};
    \addplot table[x=Iteration,y=wikilinks] {figures/stxxl_io.dat};
    \addplot table[x=Iteration,y=dbpedia] {figures/stxxl_io.dat};
    \addplot table[x=Iteration,y=twitter] {figures/stxxl_io.dat};
    \addplot table[x=Iteration,y=bsbm] {figures/stxxl_io.dat};
    \addplot table[x=Iteration,y=sp2b] {figures/stxxl_io.dat};
	\end{semilogyaxis}
\end{tikzpicture}
\caption{I/O spent on sort/scan (STXXL) for every iteration}
\label{fig:stxxl_io}
\end{subfigure}
~
\begin{subfigure}[t]{0.3\textwidth}
\raggedleft
\begin{tikzpicture}[baseline]%[baseline,trim axis left,trim axis right]
	\begin{semilogyaxis}[
tiny,
xlabel=Iteration,
ylabel=I/O on BDB (byte),
grid=major,
xtick={1,2,3,4,5,6,7,8,9,10},
height=5cm,width=5cm,
]
	\addplot table[x=Iteration,y=jamendo] {figures/bdb_io.dat};
    \addplot table[x=Iteration,y=linkedmdb] {figures/bdb_io.dat};
    \addplot table[x=Iteration,y=dblp] {figures/bdb_io.dat};
    \addplot table[x=Iteration,y=wikilinks] {figures/bdb_io.dat};
    \addplot table[x=Iteration,y=dbpedia] {figures/bdb_io.dat};
    \addplot table[x=Iteration,y=twitter] {figures/bdb_io.dat};
    \addplot table[x=Iteration,y=bsbm] {figures/bdb_io.dat};
    \addplot table[x=Iteration,y=sp2b] {figures/bdb_io.dat};
	\end{semilogyaxis}
\end{tikzpicture}
\caption{I/O spent on \s\ (BDB) for every iteration}
\label{fig:bdb_io}
\end{subfigure}
~
\begin{subfigure}[t]{0.3\textwidth}
\raggedleft
\begin{tikzpicture}[baseline]%[baseline,trim axis left,trim axis right]
	\begin{semilogyaxis}[
tiny,
xlabel=Iteration,
ylabel=Time on signature preparation (s),
grid=major,
xtick={1,2,3,4,5,6,7,8,9,10},
height=5cm,width=5cm,
]
	\addplot table[x=Iteration,y=jamendo] {figures/prepare_time.dat};
    \addplot table[x=Iteration,y=linkedmdb] {figures/prepare_time.dat};
    \addplot table[x=Iteration,y=dblp] {figures/prepare_time.dat};
    \addplot table[x=Iteration,y=wikilinks] {figures/prepare_time.dat};
    \addplot table[x=Iteration,y=dbpedia] {figures/prepare_time.dat};
    \addplot table[x=Iteration,y=twitter] {figures/prepare_time.dat};
    \addplot table[x=Iteration,y=bsbm] {figures/prepare_time.dat};
    \addplot table[x=Iteration,y=sp2b] {figures/prepare_time.dat};
	\end{semilogyaxis}
\end{tikzpicture}
\caption{Time spent on signature preparation for every iteration}
\label{fig:prepare_time}
\end{subfigure}
~
\begin{subfigure}[t]{0.3\textwidth}
\raggedleft
\begin{tikzpicture}[baseline]%[baseline,trim axis left,trim axis right]
	\begin{semilogyaxis}[
tiny,
xlabel=Iteration,
ylabel=Time on signature construction (s),
grid=major,
xtick={1,2,3,4,5,6,7,8,9,10},
height=5cm,width=5cm,
]
	\addplot table[x=Iteration,y=jamendo] {figures/signature_time.dat};
    \addplot table[x=Iteration,y=linkedmdb] {figures/signature_time.dat};
    \addplot table[x=Iteration,y=dblp] {figures/signature_time.dat};
    \addplot table[x=Iteration,y=wikilinks] {figures/signature_time.dat};
    \addplot table[x=Iteration,y=dbpedia] {figures/signature_time.dat};
    \addplot table[x=Iteration,y=twitter] {figures/signature_time.dat};
    \addplot table[x=Iteration,y=bsbm] {figures/signature_time.dat};
    \addplot table[x=Iteration,y=sp2b] {figures/signature_time.dat};
	\end{semilogyaxis}
\end{tikzpicture}
\caption{Time spent on signature construction and insertion for every iteration}
\label{fig:signature_time}
\end{subfigure}
\ref{bignamed}
~
\caption{Experiment results for Algorithm \ref{algo:build_k-bisims} for real and synthetic datasets ($k=10$)}
\label{fig:big_figure}
\end{figure*}
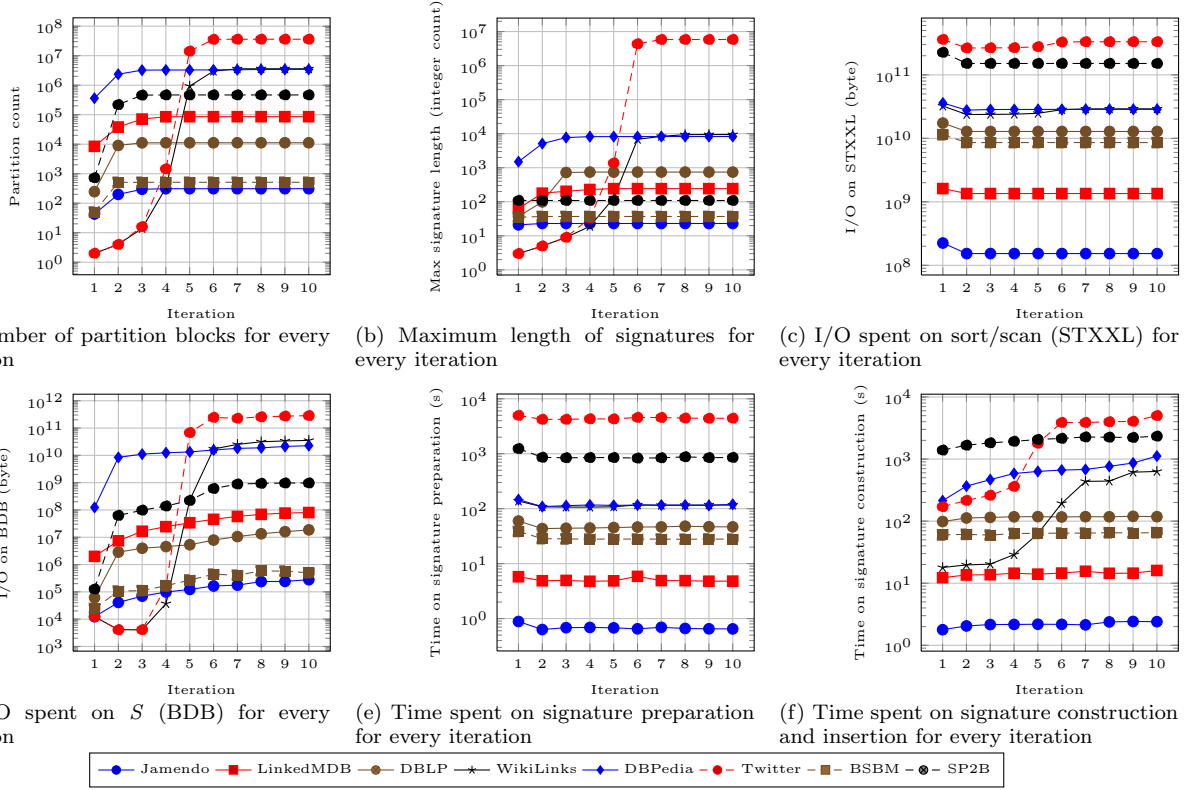

\paragraph*{Validation of implementations}

We validate the correctness of the implementation of our algorithms by comparing
their output against that of other existing solutions.  The first algorithm is the
classic bisimulation algorithm from paper \cite{smolka95}, which computes the
full bisimulation of the graph. We implement this algorithm using Python, run it on
small datasets from the Stanford Large Network Dataset Collection\footnotemark[5]
(p2p-Gnutella04, 05, 06, 08), and
 compare the output with Algorithm \ref{algo:build_k-bisims} while setting
$k$ to 100. The two algorithms produce the same partitions.

We also validate Algorithm \ref{algo:build_k-bisims} against Hellings et al.\
\cite{Hellings:2012:EEB:2213836.2213899}. Since we could handle any type of
directed graph, we can also handle acyclic graphs. We use the random DAG
generator provided along with \cite{Hellings:2012:EEB:2213836.2213899} to
generate several graphs and test them on both algorithms. They produce the same
partition results as expected.

Furthermore, we validate the algorithm \addedges\ against the algorithm
\construct.  In this experiment, for the same dataset, we first compute the
\kbisimnormal\ partitions using \construct; then split the dataset into two parts,
using the first part as the building block, and the second part as the edges to
be updated, applying \addedges\ on the second part.  Both algorithms produce the
same results.

%\eluo{
%We also run \construct\ with smaller datasets
%on a machine with 1GB main memory, but keep the buffer setting as the same.
%The experiment returns the same amount of I/Os as on bigger machines, so
%we are sure that the buffer setting works as expected, regardless of the operating system buffer effect.
%}

%\newpage
\subsection{Experiments on the localized bisimulation construction algorithm (\construct)}
\label{sec:build_experiment}

In Figure \ref{fig:big_figure} we show the experiment results for Algorithm
\ref{algo:build_k-bisims} on all datasets.  We compute the 10-bisimulation
(i.e., $k=10$) of these datasets, and measure many aspects of the running behavior
for each iteration.  Concerning time measurement, we run every experiment 5
times and take the average number. \s\ uses BerkleyDB's B-Tree index in this experiment.
\iftoggle{expand}{
  Readers can find detailed numbers from these experiments in Table \ref{table:big_table},
  found in the Appendix.
}{
%do nothing in shorter version
}

In Figure \ref{fig:partition_count}, we show the number of partition blocks every
iteration produces for all datasets.  We see that the numbers vary from one dataset to
another, where the difference is sometimes more than an order of magnitude, and
interestingly, does not directly relate to the size of the dataset.  In certain
cases (e.g., Twitter) partition size is quite large.  Moreover, many of the
datasets (e.g., Jamendo, LinkedMDB, DBLP, etc.) reach full bisimulation after 5
iterations. In fact, all datasets (including Twitter) get sufficient partition
result after 5 iterations of computation. Here we can reasonably argue that even for Twitter
dataset, the partition results after 5 iterations are too refined (e.g.,
(partition count)/(node count) $> 0.8$).

Figure \ref{fig:max_signature} shows the maximum length of signatures for each
iteration.  We observe that the signature length is usually quite short,
especially comparing with the size of the graph.
 But there are still cases
(e.g., Twitter) that the signature becomes very long (more than 1
million integers), which stresses the need for an I/O efficient solution for \s.
Note that the synthetic datasets, such as BSBM and SP2B, reach their full bisimulation partition after 3 iterations of computations,
and have rather short signatures, indicating that they are highly structured.

Figures \ref{fig:stxxl_io} and \ref{fig:bdb_io} show the I/O volume spent on
sorting/scanning (STXXL) and on interacting with \s\ (Berkeley DB).  We see for
most of the datasets, there is no dramatic change cross different iterations.  But
for Wikilinks and Twitter, the two datasets which have very few partition
blocks at the beginning and many at the end, there is a noticeable difference on \s\
for different iterations.  In this case I/O on \s\ becomes a comparable factor
with sort and scan (I/O on STXXL).

Figure \ref{fig:prepare_time} shows the time spent on preparing the signature
(line \ref{algo:build_k-bisims_preparestart} to
\ref{algo:build_k-bisims_prepareend} in Algorithm \ref{algo:build_k-bisims}) for
each iteration, which is quite stable for all datasets.  Figure
\ref{fig:signature_time} shows the time on constructing the signature and insert
into \s\ (line \ref{algo:build_k-bisim_simple_forstart} to
\ref{algo:build_k-bisim_simple_forend} in Algorithm \ref{algo:build_k-bisims}).
In this case datasets with higher degrees tend to cost more time in later
iterations, which correlate with their longer signatures and larger number of partition blocks.
For all datasets, however, the operations on constructing and looking
for signature are the dominant factor for each iteration. This brings us to
think about further optimization tasks on construction of signature and
implementation of \s.

We can conclude that the algorithm is practical to use. It can process a graph
with 100 million edges (e.g., WikiLinks and DBPedia) in under 700 seconds for one
iteration, and performance scales (almost) linearly with the number of nodes and edges.

%In Table  we show the whole big picture of the experiment results for all the datasets we mentioned before.  For each iteration, we record the running time for signature preparation and construction (average of 5 times), the I/O data volume in bytes for table scans and sorts (STXXL library) and for \s\ (Berkeley DB), and the maximum signature length of each iteration.
%We use gray color in the table to indicate where the algorithm stops (after getting the same partition counts for two iterations).

%\newpage
\subsubsection{Different implementations of \s}

As we mentioned in Section \ref{sec:more_discussion}, \s\ could be implemented in several ways.  In Figure
\ref{fig:btree_hash} we compare the overall I/O performance of \construct\ using
B-Tree and Hash indexes for \s\ on several datasets.  We see that the B-Tree
implementation slightly outperforms Hash Index for all datasets. This is most
likely due to small caching effects and locality of references during construction of
the signatures.

\usetikzlibrary{patterns}

\begin{figure}[htbp]
\centering
\begin{tikzpicture}
	\begin{semilogyaxis}[
tiny,
ylabel=Total I/O (byte),
symbolic x coords={Jamendo,LinkedMDB,DBLP,WikiLinks,DBPedia,BSBM,SP2B},
xtick = data,
x tick label style={rotate=25,anchor=east},
legend pos=north west,
height=4cm,width=8cm,
ybar,
%xlabel=Dataset,
%enlarge y limits=0.5,
grid=major,
]
	\addplot[draw=black,pattern color=blue,pattern=crosshatch dots] table[x=dataset,y=hash,] {figures/btree_hash.dat};
	\addplot[draw=black,pattern color=gray,pattern=horizontal lines] table[x=dataset,y=btree,] {figures/btree_hash.dat};
    \legend{Hash,B-Tree}
	\end{semilogyaxis}
\end{tikzpicture}
\caption{I/O comparison for B-Tree and Hash index of \s\ ($k=10$)}
\label{fig:btree_hash}

\end{figure}
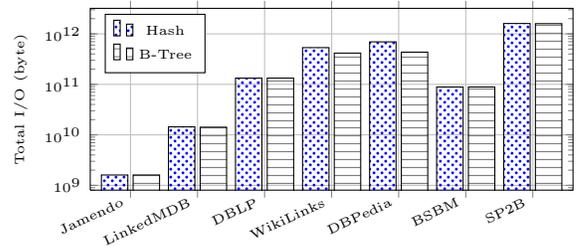

%\newpage
\subsubsection{The effect of different buffer sizes}

We allocate two buffers, one for scan and sort (STXXL buffer in our case), one for \s\ (BerkeleyDB buffer in our
case),
in order to analyze the impact of buffer size on our algorithms.
%There are two buffers in the algorithm: the buffer for table scan and sort ( and the buffer for \s\ (the Berkeley DB buffer in our case).
%In this experiment we want to see how the buffer sizes will affect the
%algorithm's performance.
To illustrate, we take the DBPedia dataset since it is
large enough to show  buffer effects.
For the sort/scan setting, we set the
buffer size ranging from 16MB to 512MB, while keeping the \s\ buffer to
128MB, recording the I/O between the buffer and the disk system. From Figure
\ref{fig:table_buffer} we see that bigger buffer does improve the
performance. But since we only gain in the external memory sorting part, a
certain amount of I/Os is inevitable for each iteration. Note that the reason
why iteration 1 has higher I/O cost is that in iteration 1
extra sorts on \nodetable\ and \edgetable\ are performed.

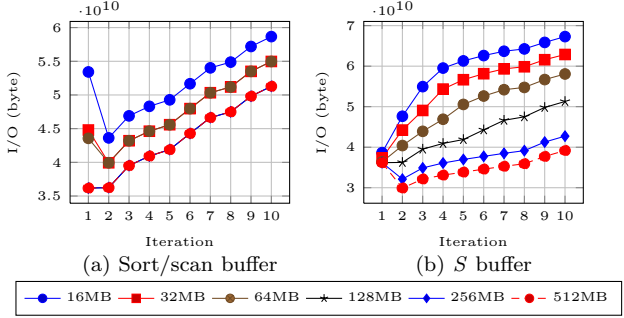
\begin{figure}[htbp]
\centering
        \begin{subfigure}[t]{0.45\columnwidth}
                \centering
\begin{tikzpicture}[baseline,trim axis left]
	\begin{axis}[
%xlabel=Iteration,
tiny,
ylabel=I/O (byte),
xlabel=Iteration,
xtick={1,2,3,4,5,6,7,8,9,10},
grid=major,
%legend pos=outer north east,
height=4cm,width=4.5cm,
legend columns=6,
legend entries={16MB,32MB,64MB,128MB,256MB,512MB},
legend to name=named_bdbs,
]
	\addplot table[x=Iteration,y=16] {figures/table_buffer.dat};
    \addplot table[x=Iteration,y=32] {figures/table_buffer.dat};
    \addplot table[x=Iteration,y=64] {figures/table_buffer.dat};
    \addplot table[x=Iteration,y=128] {figures/table_buffer.dat};
    \addplot table[x=Iteration,y=256] {figures/table_buffer.dat};
    \addplot table[x=Iteration,y=512] {figures/table_buffer.dat};
    %\legend{16MB,32MB,64MB,128MB,256MB,512MB}
	\end{axis}
\end{tikzpicture}
%\caption{}
\caption{Sort/scan buffer}
\label{fig:table_buffer}
        \end{subfigure}%
        ~%add desired spacing between images, e. g. ~, \quad, \qquad etc.
        \begin{subfigure}[t]{0.45\columnwidth}
                \centering
\begin{tikzpicture}[baseline,trim axis left]
	\begin{axis}[
%xlabel=Iteration,
tiny,
ylabel=I/O (byte),
xlabel=Iteration,
xtick={1,2,3,4,5,6,7,8,9,10},
grid=major,
%legend pos=outer north east,
height=4cm,width=4.5cm,
]
	\addplot table[x=Iteration,y=16] {figures/s_buffer.dat};
    \addplot table[x=Iteration,y=32] {figures/s_buffer.dat};
    \addplot table[x=Iteration,y=64] {figures/s_buffer.dat};
    \addplot table[x=Iteration,y=128] {figures/s_buffer.dat};
    \addplot table[x=Iteration,y=256] {figures/s_buffer.dat};
    \addplot table[x=Iteration,y=512] {figures/s_buffer.dat};
    %\legend{16MB,32MB,64MB,128MB,256MB,512MB}
	\end{axis}
\end{tikzpicture}
\caption{\s\ buffer}
\label{fig:s_buffer}
        \end{subfigure}
\ref{named_bdbs}
~
\caption{I/O for different buffer size setting for sort/scan and \s\ ($k=10$)}
\end{figure}

For the setting on \s, we set the buffer size ranging from 16MB to 512MB,
while keeping the sort/scan buffer to be 128MB, recording the I/O of the buffer to
the disk system. From Figure \ref{fig:s_buffer} we also see that more buffer
brings less I/O, as expected. However, in this case the buffer size change has a
bigger impact on the I/O performance. This indicates that if we have a certain
amount of memory space, it is more beneficial to allocate more memory to the
\s\ buffer than to the sort/scan buffer. Note that \s\ buffer
also shows quite high hit ratio during execution (more than 0.98 for DBPedia in
all settings).

%\subsubsection{Algorithm performance with different dataset sizes}
\subsubsection{Scalability}

In order to measure how well the algorithm scales, we generate different size of
SP2B datasets (edge count 1M, 5M, 10M, 50M, 100M, 500M), and measure the I/O and
elapsed time for each dataset. In Figure \ref{fig:per_edge} we see that the
time spent on each edge is on the order of $10^{-5}$ seconds, and the I/O spent on each edge
is under 4000 bytes (which is one typical disk page size). The algorithm's
performance scales (almost) linearly with the data size.

\begin{figure}[htbp]
\centering
\begin{subfigure}[b]{0.45\columnwidth}
\centering
\begin{tikzpicture}
	\begin{semilogxaxis}[tiny,xlabel=Data size (number of edges),ylabel=Time (s) per edge,grid=major,height=4cm,width=4.5cm]
	\addplot table[x=size,y=timeonedge] {figures/per_edge.dat};
	\end{semilogxaxis}
\end{tikzpicture}
\end{subfigure}
~
\begin{subfigure}[b]{0.45\columnwidth}
\centering
\begin{tikzpicture}
	\begin{semilogxaxis}[tiny,xlabel=Data size (number of edges),ylabel=I/O (byte) per edge,grid=major,height=4cm,width=4.5cm]
	\addplot table[x=size,y=ioonedge] {figures/per_edge.dat};
	\end{semilogxaxis}
\end{tikzpicture}
\end{subfigure}
\caption{Time and I/O spent on each edge on average ($k=10$)}
\label{fig:per_edge}
\end{figure}
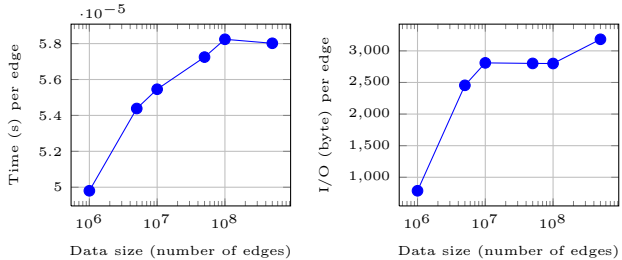

\subsection{Experiments on the edge update algorithm (\addedges)}

Edge updates are common operations for graph data. For our datasets, adding one
edge means to add a link between two wiki pages (WikiLinks), to add more
information to one publication or author (DBLP), to follow one more person
(Twitter) and so on. Sometimes we would like to also add several edges together
at once.  So in this subsection we test the performance of Algorithm
\ref{algo:add_edge} (\addedges), first adding a single edge and then adding a set of edges.

\subsubsection{Observations on single edge update}

To create the dataset for testing, we randomly take one edge from the edge set,
perform \construct\ on the rest of the dataset, and apply \addedges\ on this
edge. We believe the edge selection is more natural this way, since it take into
account the distribution of edges among nodes. We repeat the experiment 10 times
and take the average of the measured numbers. In Figure \ref{fig:check_node} we
show how many nodes are checked for adding one edge to the graph in each
iteration. In Figure \ref{fig:change_node} we show how many nodes actually
change their partition IDs in each iteration. From the figures we see that the
behavior varies for different datasets; graphs that have larger degrees tend to
propagate more changes to later iterations, which complies with our intuition.

\begin{figure}[htbp]
%\centering
\begin{subfigure}[t]{0.23\textwidth}
\centering
\begin{tikzpicture}
	\begin{semilogyaxis}[
tiny,
xlabel=Iteration,
ylabel=Check number of nodes,
grid=major,
ytick={0.1,1,10,100,1000,10000,100000,1000000,10000000},
xtick={1,2,3,4,5,6,7,8,9,10},
%legend pos=outer north east,
height=4cm,width=4.5cm,
%legend style={font=\footnotesize}
%legend columns=1,
%legend entries={Jamendo,LinkedMDB,DBLP,WikiLinks,DBPedia,BSBM,SP2B},
%legend to name=named,
]
	\addplot table[x=Iteration,y=jamendo] {figures/check_node.dat};
    \addplot table[x=Iteration,y=linkedmdb] {figures/check_node.dat};
    \addplot table[x=Iteration,y=dblp] {figures/check_node.dat};
    \addplot table[x=Iteration,y=wikilinks] {figures/check_node.dat};
    \addplot table[x=Iteration,y=dbpedia] {figures/check_node.dat};
    \addplot table[x=Iteration,y=bsbm] {figures/check_node.dat};
    \addplot table[x=Iteration,y=sp2b] {figures/check_node.dat};
    %\legend{Jamendo,LinkedMDB,DBLP,WikiLinks,DBPedia,BSBM,SP2B}
	\end{semilogyaxis}
\end{tikzpicture}
\caption{Number of nodes being checked for each iteration}
\label{fig:check_node}
\end{subfigure}
~
\begin{subfigure}[t]{0.23\textwidth}
\centering
\begin{tikzpicture}
	\begin{semilogyaxis}[
tiny,
xlabel=Iteration,
ylabel=Change number of nodes,
grid=major,
ytick={0.1,1,10,100,1000,10000,100000,1000000,10000000},
xtick={1,2,3,4,5,6,7,8,9},
%legend pos=outer north east,
height=4cm,width=4.5cm,
]
	\addplot table[x=Iteration,y=jamendo] {figures/change_node.dat};
    \addplot table[x=Iteration,y=linkedmdb] {figures/change_node.dat};
    \addplot table[x=Iteration,y=dblp] {figures/change_node.dat};
    \addplot table[x=Iteration,y=wikilinks] {figures/change_node.dat};
    \addplot table[x=Iteration,y=dbpedia] {figures/change_node.dat};
    \addplot table[x=Iteration,y=bsbm] {figures/change_node.dat};
    \addplot table[x=Iteration,y=sp2b] {figures/change_node.dat};
    %\legend{Jamendo,LinkedMDB,DBLP,WikiLinks,DBPedia,BSBM,SP2B}
	\end{semilogyaxis}
\end{tikzpicture}
\caption{Number of nodes that change partitions for each iteration}
\label{fig:change_node}
\end{subfigure}
~
\begin{subfigure}[t]{0.23\textwidth}
\centering
\begin{tikzpicture}
	\begin{semilogyaxis}[
tiny,
xlabel=Iteration,
ylabel=Change number of partitions,
ymin=0,
grid=major,
legend style={at={(2.3,1)}, anchor=north east},
ytick={0.1,1,10,100,1000,10000,100000,1000000,10000000},
xtick={1,2,3,4,5,6,7,8,9,10},
%legend pos=outer north east,
height=4cm,width=4.5cm,
]
	\addplot table[x=Iteration,y=jamendo] {figures/change_partition.dat};
    \addplot table[x=Iteration,y=linkedmdb] {figures/change_partition.dat};
    \addplot table[x=Iteration,y=dblp] {figures/change_partition.dat};
    \addplot table[x=Iteration,y=wikilinks] {figures/change_partition.dat};
    \addplot table[x=Iteration,y=dbpedia] {figures/change_partition.dat};
    \addplot table[x=Iteration,y=bsbm] {figures/change_partition.dat};
    \addplot table[x=Iteration,y=sp2b] {figures/change_partition.dat};
    \legend{Jamendo,LinkedMDB,DBLP,WikiLinks,DBPedia,BSBM,SP2B}
	\end{semilogyaxis}
\end{tikzpicture}
\caption{Number of partitions whose members change their partition blocks for each iteration}
\label{fig:change_partition}
\end{subfigure}
~
\begin{subfigure}[t]{0.23\textwidth}
%%\ref{named}
\end{subfigure}
\caption{Experiment on \addedges\ when $k=10$}%, showing that how many nodes are checked, how many nodes' partition identifiers are changed, and how many partitions are changed in each iteration.}
\end{figure}
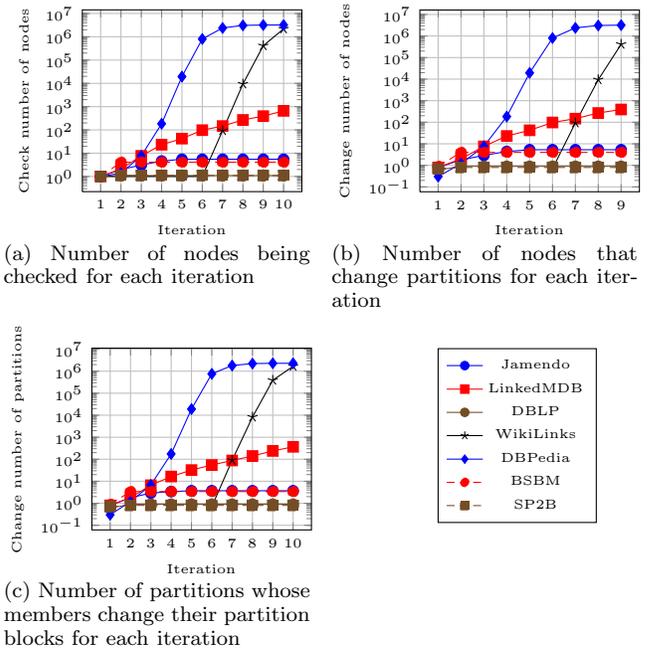

%\newpage
Since there is a chance that many nodes are changed but they may all belong to a certain set of partitions, we also show how many partitions change their members in each iteration (Figure \ref{fig:change_partition}).
We see that the behavior of Figure \ref{fig:change_partition} is closely related to that of Figure \ref{fig:change_node}.
%From  we see that the number of partitions  the number of nodes that changed.
\subsubsection{Comparison of \construct\ and \addedges\ (single edge update)}
\label{sec:compare_con_add}

After edge insertion, if there is no update algorithm available, the only choice
to get the \kbisimnormal\ partition is to execute the \construct\ from scratch on
the new dataset. So this would be the baseline for the \addedges\ algorithm to
compare.  In the following we compare the overall I/O and time (Figure
\ref{fig:build_update_iotime}) of the two algorithms. We see that indeed the
\addedges\ algorithm always achieves a better performance than using \construct\
to recompute the \kbisimnormal\ partition result from scratch, with up to an
order of magnitude improvement.

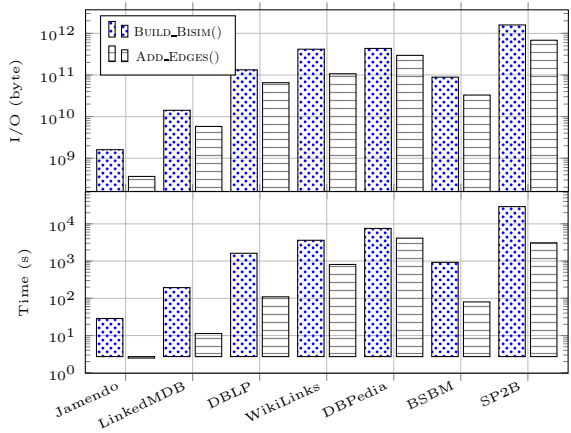
\begin{figure}[htbp]
\centering
\begin{tikzpicture}
\begin{groupplot}[
    group style={
        group name=my plots,
        group size=1 by 2,
        xlabels at=edge bottom,
        xticklabels at=edge bottom,
        vertical sep=0pt
    },
    ybar,
    tiny,
    xtick = data,
    width=8cm,
    height=4cm,
    symbolic x coords={Jamendo,LinkedMDB,DBLP,WikiLinks,DBPedia,BSBM,SP2B},
    x tick label style={rotate=25,anchor=east},
    legend pos=north west,
    grid=major,
    xtick align=outside
]
\nextgroupplot[ymode=log,ylabel=I/O (byte)]
\addplot[draw=black,pattern color=blue,pattern=crosshatch dots] table[x=dataset,y=Construction,] {figures/build_update_io.dat};
\addplot[draw=black,pattern color=gray,pattern=horizontal lines] table[x=dataset,y=Maintenance,] {figures/build_update_io.dat};
\legend{\construct,\addedges}
\nextgroupplot[ymode=log,ylabel=Time (s)]
	\addplot[draw=black,pattern color=blue,pattern=crosshatch dots] table[x=dataset,y=Construction,] {figures/build_update_time.dat};
	\addplot[draw=black,pattern color=gray,pattern=horizontal lines] table[x=dataset,y=Maintenance,] {figures/build_update_time.dat};
\end{groupplot}
\end{tikzpicture}
\caption{I/O and time comparison for \construct\ and \addedges\ after inserting one edge to the dataset ($k=10$)}
\label{fig:build_update_iotime}
\end{figure}

\subsubsection{Comparison of \construct\ and \addedges\ in extreme cases (single edge update)}

From the above experiments, we see that the performance of the algorithms are
highly related to the datasets they process.  For some datasets, the update
algorithm is very much favorable while in other cases not so much.  In the
following, we would like to gain a better understanding of this phenomena.

We achieve this with two synthetic datasets, triggering both the extreme cases
where the construction algorithm benefits the most and the update algorithm benefits
the most.  The first dataset, Dbest, shows a best-case scenario that the update algorithm
can achieve relative to the construction algorithm. In this case we create a
full k-ary tree, with edges pointing from parents to their children. When adding
one edge to the tree, we add one edge to the leaf node, so that no node's
signature would change after the insertion. In this case the update algorithm
does the least amount of work, without propagating any change to further
iterations during execution. Figure \ref{fig:example_dbest} shows an example of
Dbest, which is a binary tree with height 3. The dashed edge is the newly added
edge.

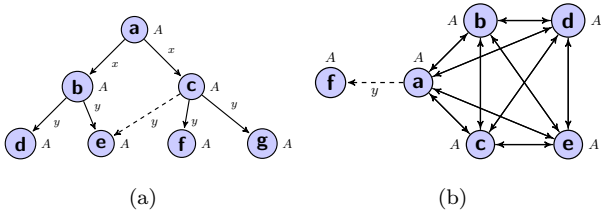
\begin{figure}[htbp]
\centering
\begin{subfigure}[b]{0.45\columnwidth}
\centering
\resizebox{4cm}{!}{
\begin{tikzpicture}[->,>=stealth',shorten >=1pt,auto,node distance=2cm,
  thick,main node/.style={circle,fill=blue!20,draw,font=\sffamily\Large\bfseries}]
\centering
%\useasboundingbox (-3,-3) rectangle (4,0.5);
  \node[main node] (1) [label=right:$A$]{a};
  \node[main node] (2) [label=right:$A$,below left of=1] {b};
  \node[main node] (3) [label=right:$A$,below right of=1] {c};
  \node[main node] (4) [label=right:$A$,below left of=2] {d};
  \node[main node] (5) [label=right:$A$,right of=4] {e};
  \node[main node] (6) [label=right:$A$,right of=5] {f};
  \node[main node] (7) [label=right:$A$,right of=6] {g};

  \path[every node/.style={font=\sffamily\small}]
    (1) edge node {$x$} (2)
    (1) edge node {$x$} (3)
    (2) edge node {$y$} (4)
    (2) edge node {$y$} (5)
    (3) edge node {$y$} (6)
    (3) edge node {$y$} (7);
  \draw[dashed] (3) edge node {$y$} (5);
  %\draw (current bounding box.south west) rectangle (current bounding box.north east);
\end{tikzpicture}
}
\caption{}%\caption{Example for Dbest}
\label{fig:example_dbest}
\end{subfigure}
~
\begin{subfigure}[b]{0.45\columnwidth}
\centering
\resizebox{4cm}{!}{
\begin{tikzpicture}[->,>=stealth',shorten >=1pt,auto,node distance=2cm,
  thick,main node/.style={circle,fill=blue!20,draw,font=\sffamily\Large\bfseries}]
\centering
  \node[main node] (1) [label=above:$A$]{a};
  \node[main node] (2) [label=left:$A$,above right of=1] {b};
  \node[main node] (3) [label=left:$A$,below right of=1] {c};
  \node[main node] (4) [label=right:$A$,right of=2] {d};
  \node[main node] (5) [label=right:$A$,right of=3] {e};
  \node[main node] (6) [label=above:$A$,left of=1] {f};

  \path[every node/.style={font=\sffamily\small}]
    (1) edge node {}(2)
    (1) edge node {}(3)
    (1) edge node {}(4)
    (1) edge node {}(5)
    (2) edge node {}(1)
    (2) edge node {}(3)
    (2) edge node {}(4)
    (2) edge node {}(5)
    (3) edge node {}(1)
    (3) edge node {}(2)
    (3) edge node {}(4)
    (3) edge node {}(5)
    (4) edge node {}(1)
    (4) edge node {}(2)
    (4) edge node {}(3)
    (4) edge node {}(5)
    (5) edge node {}(1)
    (5) edge node {}(2)
    (5) edge node {}(3)
    (5) edge node {}(4);

%  \path[every node/.style={font=\sffamily\small}]
%    (1) edge node {$x$} (2)
%    (1) edge node {$x$} (3)
%    (1) edge node {$x$} (4)
%    (1) edge node {$x$} (5)
%    (2) edge node {$x$} (1)
%    (2) edge node {$x$} (3)
%    (2) edge node {$x$} (4)
%    (2) edge node {$x$} (5)
%    (3) edge node {$x$} (1)
%    (3) edge node {$x$} (2)
%    (3) edge node {$x$} (4)
%    (3) edge node {$x$} (5)
%    (4) edge node {$x$} (1)
%    (4) edge node {$x$} (2)
%    (4) edge node {$x$} (3)
%    (4) edge node {$x$} (5)
%    (5) edge node {$x$} (1)
%    (5) edge node {$x$} (2)
%    (5) edge node {$x$} (3)
%    (5) edge node {$x$} (4);
  \draw[dashed] (1) edge node {$y$} (6);
  %\draw (current bounding box.south west) rectangle (current bounding box.north east);
\end{tikzpicture}
}
\caption{}%\caption{Example for Dworst}
\label{fig:example_dworst}
\end{subfigure}
\caption{Examples for Dbest (\ref{fig:example_dbest}) and Dworst (\ref{fig:example_dworst}) datasets}
\end{figure}

The second dataset, Dworst, exhibits a worst-case scenario for the update
algorithm, relative to construction. In this case we create a
complete graph, with edges all labeled with \emph{x}. Then when adding one more
edge (labeled \emph{y}) to one of the nodes, every other node in each iteration
is affected and therefore all the nodes' signatures are changed. The update
algorithm has to check all nodes in every iteration. Figure
\ref{fig:example_dworst} shows an example of Dworst, a complete graph with 5
nodes. The dashed edge is the newly added edge.

\begin{figure}[htbp]
\centering
\begin{subfigure}[b]{0.45\columnwidth}
\centering
\begin{tikzpicture}
	\begin{axis}[
tiny,
xlabel=Iteration,ylabel=I/O (byte),
grid=major,
xtick={1,2,3,4,5,6,7,8,9,10},
legend columns=2,
legend entries={Dbest \construct,Dbest \addedges,Dworst \construct,Dworst \addedges},
legend to name=named_bestworst,
%legend pos=outer north east,
height=4cm,width=4.5cm,
]
	\addplot[mark=*,red,solid,mark size=2] table[x=Iteration,y=bestconstruct] {figures/best_worst_io.dat};
    \addplot[mark=square*,blue,solid,mark size=2] table[x=Iteration,y=bestupdate] {figures/best_worst_io.dat};
    \addplot[mark=triangle*,brown,solid,mark size=2] table[x=Iteration,y=worstconstruct] {figures/best_worst_io.dat};
    \addplot[mark=o,black,solid,mark size=2] table[x=Iteration,y=worstupdate] {figures/best_worst_io.dat};
    %\legend{Dbest \construct,Dbest \addedge,Dworst \construct,Dworst \addedge}
	\end{axis}
\end{tikzpicture}
%\caption{I/O}
%\label{fig:dbest}
\end{subfigure}
~
\begin{subfigure}[b]{0.45\columnwidth}
\centering
\begin{tikzpicture}
	\begin{axis}[
tiny,
xlabel=Iteration,ylabel=Time (s),
grid=major,
xtick={1,2,3,4,5,6,7,8,9,10},
ytick={0,100,200,300,400},
%legend pos=outer north east,
height=4cm,width=4.5cm,
]
	\addplot[mark=*,red,solid,mark size=2] table[x=Iteration,y=bestconstruct] {figures/best_worst_time.dat};
    \addplot[mark=square*,blue,solid,mark size=2] table[x=Iteration,y=bestupdate] {figures/best_worst_time.dat};
    \addplot[mark=triangle*,brown,solid,mark size=2] table[x=Iteration,y=worstconstruct] {figures/best_worst_time.dat};
    \addplot[mark=o,black,solid,mark size=2] table[x=Iteration,y=worstupdate] {figures/best_worst_time.dat};
    \end{axis}
\end{tikzpicture}
%\caption{Time}
%\label{fig:dworst}
\end{subfigure}
\ref{named_bestworst}
\caption{Time and I/O comparison for Dbest and Dworst by applying \construct\ and \addedges\ algorithms on both ($k=10$)}
\label{fig:dbestio}
\end{figure}
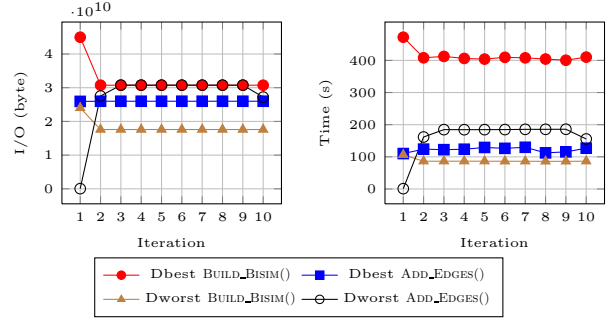

We generate Dbest and Dworst on the scale of 100 million edges, and measure the
elapsed time and I/O costs (Figure \ref{fig:dbestio}) for both the construction
(\construct) and edge update (\addedges) algorithms in each iteration. We see that
indeed for Dbest, the update algorithm shows a 4 times speed-up in time
compared with the construction algorithm. For Dworst, the update algorithm is 2
times slower in time than the construction algorithm.

%\newpage
\subsubsection{Experiment on multiple edges update}
\label{sec:add_edges}
To test the performance of multiple edges update, we randomly select a set of edges from the dataset (edge count 1, 10, 100, \dots, 1M), and apply the algorithm \addedges\ upon them,
recording the I/O and elapsed time performances.
In Figure \ref{fig:batch_timeio}, we show the I/O improvement ratio and time speed up ratio (both construct/update) for all cases (taking the average).
%Data points of the same shapes but with bigger sizes mean that from that point it is better to switch to the construction algorithm (\construct).
A gray line is drawn at $y=1$ for both figures to split the space to indicate whether \addedges\ performs better than \construct\ or not.
From the figure we see that
%indeed update multiple edges would make the cost amortized to each edge much cheaper.
for many of the datasets, it is beneficial to do batch update (\addedges) up until $10^4$ edges.
An order of magnitude time speed up is observed for Jamendo, LinkedMDB and DBLP.
In fact, if we consider the time cost for Jamendo and DBLP, it is always favorable to use \addedges\ in all cases.
For dataset DBPedia, however, changes propagate rapidly in the first few iterations, therefore the construction algorithm (\construct) becomes a better choice when there are more than ten edges to be updated.

\begin{figure}[htbp]
%\centering
\begin{subfigure}[b]{0.45\columnwidth}
\centering
\begin{tikzpicture}
	\begin{loglogaxis}[
tiny,
xlabel=Update batch size,
ylabel=I/O improvement ratio,
legend columns=5,
legend entries={Jamendo,LinkedMDB,DBLP,WikiLinks,DBPedia},
legend to name=named_batch,
log basis y=2,
grid=none,
yminorgrids=true,
%legend pos=outer north east,
height=5cm,width=4.5cm,
]

	\addplot[mark=*,red,solid,mark size=2] table[x=size,y=jamendo,] {figures/batch_io_ratio.dat};
    \addplot[mark=square*,blue,solid,mark size=2] table[x=size,y=linkedmdb,] {figures/batch_io_ratio.dat};
    \addplot[mark=triangle*,brown,solid,mark size=2] table[x=size,y=dblp,] {figures/batch_io_ratio.dat};
    \addplot[mark=diamond*,cyan,solid,mark size=2] table[x=size,y=wikilinks,] {figures/batch_io_ratio.dat};
    \addplot[mark=o,black,solid,mark size=2] table[x=size,y=dbpedia,] {figures/batch_io_ratio.dat};
    \draw [gray,rounded corners] (axis cs:0.1,0.1) rectangle (axis cs:10000000,1);
    \node at (axis cs:40,0.67){\tiny{$ \begin{array}{l}\mbox{cost more} \\ \mbox{than}~\construct\end{array}$}};
    \node at (axis cs:60000,3){\tiny{$ \begin{array}{l}\mbox{cost less than} \\ ~\construct\end{array}$}};

	\end{loglogaxis}
\end{tikzpicture}
\end{subfigure}
~
\begin{subfigure}[b]{0.45\columnwidth}
\centering
\begin{tikzpicture}
    \begin{loglogaxis}[
	%\begin{semilogxaxis}[
tiny,
xlabel=Update batch size,
ylabel=Time speed up ratio,
log basis y=2,
grid=none,
height=5cm,width=4.5cm,
]

	\addplot[mark=*,red,solid,mark size=2] table[x=size,y=jamendo,] {figures/batch_time_ratio.dat};
    \addplot[mark=square*,blue,solid,mark size=2] table[x=size,y=linkedmdb,] {figures/batch_time_ratio.dat};
    \addplot[mark=triangle*,brown,solid,mark size=2] table[x=size,y=dblp,] {figures/batch_time_ratio.dat};
    \addplot[mark=diamond*,cyan,solid,mark size=2] table[x=size,y=wikilinks,] {figures/batch_time_ratio.dat};
    \addplot[mark=o,black,solid,mark size=2] table[x=size,y=dbpedia,] {figures/batch_time_ratio.dat};
    \draw [gray,rounded corners] (axis cs:0.1,0.1) rectangle (axis cs:10000000,1);
    \node at (axis cs:20,0.4){\tiny{$ \begin{array}{l}\mbox{slower than} \\ \construct\end{array}$}};
    \node at (axis cs:100,3){\tiny{$ \begin{array}{l}\mbox{faster than} \\ \construct\end{array}$}};
	%\end{semilogxaxis}
\end{loglogaxis}
\end{tikzpicture}
\end{subfigure}
\centering
\ref{named_batch}
\caption{I/O (left) and time (right)  improvement ratio $\frac{cost(\construct)}{cost(\addedges)}$ for batch edge updates ($k=10$)}
\label{fig:batch_timeio}
\end{figure}
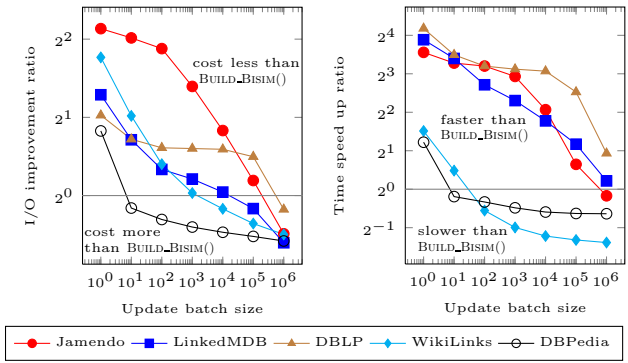

\vspace{-5pt}

%\newpage
\section{Conclusion and future work}\label{sec:conclude}

In this paper we have presented, to our knowledge, the first I/O efficient
general-purpose algorithms for constructing and maintaining localized
bisimulation partitions on massive disk-resident graphs.
%The I/O cost of the
%construction algorithm is bounded by $O(k\cdot sort(|\et|) + k\cdot scan(|\nt|)
%+ sort(|\nt|))$, while the maintenance algorithms are bounded by $O(k\cdot
%sort(|\et|) + k\cdot sort(|\nt|))$.
A theoretical analysis showed, and an extensive empirical study confirmed,
that our algorithms are not only efficient and practical to use, but also scale
well with the size of the data.

We close by listing a few promising research directions for further study.
First, it would be interesting to
explore adaptations and extensions of our algorithms for alternative hardware platforms
(e.g., multicore, SSD).  Second, as we indicated at various points, many
alternative data structures and join algorithms can be investigated for
optimizing various aspects of the proposed algorithms.  
Third, because of their bulk streaming-based nature, many aspects of our algorithms naturally lend themselves to state-of-the-art parallel and distributed computing frameworks such as MapReduce.  Studying the possibilities for leveraging our solutions to further scale the performance of these frameworks on real world graphs is certainly an interesting research direction.
%Third, many aspects of
%our algorithms naturally lend themselves to parallel or distributed solutions
%(e.g., combining our solutions with distributed computing frameworks such as MapReduce to scale their performance on realworld graphs);
%this
%is certainly an interesting direction for further study.
Last but not least,
the ideas developed in this paper provide a basis for investigating
related problems such as computing and maintaining \emph{simulation} partitions in
external memory (e.g., \cite{Gentilini2003}).

%Parallel/distributed solutions upon our work, or make use of the new hardware .
%Adaptive algorithm that can take care of the memory model (in memory/external memory/distributed).

\balance

%\iftoggle{expand}{
\paragraph*{Acknowledgments}
The research of YW is supported by Research Foundation Flanders (FWO) during her sabbatical visit to Hasselt University, Belgium.
The research of YL, GF, JH and PD is supported by the Netherlands Organisation
for Scientific Research (NWO).
%\eluo{do we thank the readers here?}
%}{
%do nothing in shorter version
%}

%the reference
\bibliographystyle{abbrv}
%\small
\bibliography{paper2012}  % vldb_sample.bib is the name of the Bibliography in this case
% latex bibtex latex
% disable for submission
\iftoggle{expand}{
   %\newpage
\begin{appendix}
\section{Proofs}
\subsection{Proofs for the body}

%%%%%%%%%%%%%%%%%%%%%%%%%%%%%%%%%%%%%%%%%%%%%%%%%%%%%%%%%%%%%%%%%%%%%%%%%%%%%%%%%%
\begin{proposition}
\label{prop:1}
$u \approx^k v \Rightarrow u \approx^{k-1} v~(k > 0)$.
%\iftoggle{expand}{
%  Proof can be found at Appendix \ref{proof:1}.
%}{
%Proof omitted.
%%do nothing in shorter version
%}
\end{proposition}
%%%%%%%%%%%%%%%%%%%%%%%%%%%%%%%%%%%%%%%%%%%%%%%%%%%%%%%%%%%%%%%%%%%%%%%%%%%%%%%%%%
\begin{proof}[for Proposition \ref{prop:1}]
\label{proof:1}
By induction on $k$.

(1) $k=1$. This is obvious, as $0$-bisimilarity just enforces equality of node
labels.

(2) $k>1$. Assume that this holds for $j-1$ ($\approx^{j-1} \Rightarrow
\approx^{j-2}, 0<j-1<k$), we want to show that this also holds for $j$ ($u
\approx^j v \Rightarrow u \approx^{j-1} v$).  Let $u \approx^{j} v$.  According
to the definition, for all outgoing edges $(u,u') \in E$, there exists some edge
$(v,v') \in E$, such that $u' \approx^{j-1} v'$ and $\edgeLabel{u,u'} =
\edgeLabel{v,v'}$, and vice versa.  Since $\approx^{j-1} \Rightarrow
\approx^{j-2}$, we have $u' \approx^{j-2} v'$, then we have $u \approx^{j-1} v$.
So $u \approx^j v \Rightarrow u \approx^{j-1} v$.
\end{proof}

\begin{proof}[for Proposition \ref{prop:equal}]
\label{proof:equal}
$\Rightarrow$:

(1) For $k=0$, this is trivial, since $\pid{0}{u}=\pid{0}{v}$.

(2) For $k > 0$,
(which also means $u \approx^k v$), we want to show that $\sig{k}{u}=\sig{k}{v}$.
According to Proposition \ref{prop:1}, $u \approx^k v \Rightarrow u \approx^0
v$, so that $\pid{0}{u}=\pid{0}{v}$. And for each outgoing edge $(u,u')$ of $u$,
there exists some outgoing edge $(v,v')$ of $v$, such that $u' \approx^{k-1}
v'$, then $\pid{k-1}{u'}=\pid{k-1}{v'}$, and $\edgeLabel{u,u'} =
\edgeLabel{v,v'}$.  Therefore each pair in $\sig{k}{u}$ equals to some pair in
$\sig{k}{v}$, and vice versa.  Then we have $\sig{k}{u}=\sig{k}{v}$.

$\Leftarrow$:

(1) For $k=0$, this is obvious.

(2) For $k > 0$.  Let $\sig{k}{u}=\sig{k}{v}$, we want to show that
$\pid{k}{u}=\pid{k}{v}$ (or $u \approx^k v$). Since $\sig{k}{u}=\sig{k}{v}$, we know
that for every outgoing edge $(u,u')$ of $u$, we have a pair
$(\edgeLabel{u,u'},\pid{k-1}{u'})$ in $\sig{k}{u}$, we can find an equal pair
$(\edgeLabel{v,v'},\pid{k-1}{v'})$ in $\sig{k}{v}$, such that
$\pid{k-1}{u'}=\pid{k-1}{v'}$ and $\edgeLabel{u,u'} = \edgeLabel{v,v'}$.
By definition, this means $u \approx^k v$.  Then we have $\pid{k}{u}=\pid{k}{v}$.
\end{proof}

%%%%%%%%%%%%%%%%%%%%%%%%%%%%%%%%%%%%%%%%%%%%%%%%%%%%%%%%%%%%%%%%%%%%%%%%%%%%%%%%%%

\begin{proof}[for Theorem \ref{theory:construct}]
\label{proof:construct}

After all the I/O cost of one iteration of \kbisim\ computation is bounded by
$O(sort(|\et|)) + O(scan(|\nt|))$, $k$ is a given input, and there is one extra sort on \nodetable\ in iteration 1. Hence Algorithm
\ref{algo:build_k-bisims} has the I/O complexity of $O(k\cdot sort(|\et|) + k\cdot scan(|\nt|) + sort(|\nt|))$.

During computation, only one \nt\ and \et\ are used, and \s\ is used.
The space upper bound for \s\ is the same as the space upper bound for all signatures.
Since in the algorithm, we construct all signatures by joining the information from \nt\ and \emph{F} (which is a projection of \et),
the space upper bound of \s\ is O(|\nt|+|\et|). Therefore, the overall space complexity upper bound of Algorithm \ref{algo:build_k-bisims}
is O(|\nt|+|\et|).

We prove correctness inductively.

(1) $k=0$. Since we are following the definition, this is obvious.

(2) $k>0$. Assume we get the correct $(k-1)$ \emph{bisimulation} partitioning
results. In iteration $k$, for each node $u$ in \nodetable, we construct
\sig{k}{u} and insert it in \s\ to get \pid{k}{u}. According to Proposition
\ref{prop:equal} and the definition of \s, we are sure that \pid{k}{u} is
correct.
\end{proof}

%%%%%%%%%%%%%%%%%%%%%%%%%%%%%%%%%%%%%%%%%%%%%%%%%%%%%%%%%%%%%%%%%%%%%%%%%%%%%%%%%%

\begin{proof}[for Theorem \ref{theory:update}]
\label{proof:update}
After all the I/O cost of one iteration of Algorithm \ref{algo:add_edge} is bounded by $O(sort(|\et|)) + O(sort(|\nt|))$, and the upper bound of the number of iterations is $k$. Hence Algorithm \ref{algo:add_edge} has the given I/O complexity.

During computation, only one \nt\ and \et\ are used, and \s\ is used.
Here the node table contains historical information from iteration 0 to k, so comparing with the original \nt, the space upper bound is $O(k\cdot|\nt|)$.
Also according to the algorithm, every iteration would have to save its signature mapping to \s, so the space upper bound of \s\ is $O(k\cdot|\et|)$.
Therefore, the overall space complexity upper bound of Algorithm \ref{algo:add_edge}
is $O(k\cdot|\nt|+k\cdot|\et|)$.

Let $(s, l, t)$ be the new edge.  After we insert $s,t$ to $N$, \pid{0}{u} will
not change for any $u\in N$.  So, according to Definition \ref{def:sig}, there
are  only two
ways that \sig{j}{u} $(0<j\le k)$ could be affected:

(1) a new pair $(\edgeLabel{v},\pid{j-1}{v})$ appears, or

(2) changes of \pid{j-1}{v} in some existing pair $(\edgeLabel{v},\allowbreak \pid{j-1}{v})$, where $v$ is some child of $u$.

Case (1) can only be caused by adding a new edge to $u$, so that in our case
this can only happen to \sig{j}{s} $(0<j\le k)$, and we capture these changes in
line \ref{algo:add_edge_oldsig} of Algorithm \ref{algo:add_edge}.  The second
case can only happen when the $pId_{j-1}$ for the children of $u$ changes. We
capture (and propagate) these changes in line \ref{algo:add_edge_propagate} of
Algorithm \ref{algo:add_edge}.  Therefore, we capture all changes in the
signatures of $u \in N$, and recompute the signatures accordingly. Hence
Algorithm \ref{algo:add_edge} produces the correct \kbisim\ partitioning result.
\end{proof}

%%%%%%%%%%%%%%%%%%%%%%%%%%%%%%%%%%%%%%%%%%%%%%%%%%%%%%%%%%%%%%%%%%%%%%%%%%%%%%%

\subsection{Alternative definitions for localized bisimulation}
In this section, we show the equivalence of various definitions of localized bisimulation that are studied in the literature.
We have an alternative definition for $k$-\emph{bisimilar}
\cite{HeY04,Kaushik2002,QunLO03}:
\begin{definition}
\label{def:kaushik}
Let $k\geq 0$ and \graph\ be a graph.
Nodes $u,v \in N$ are called {\em Kaushik $k$-bisimilar} (denoted as
$u~\hat{\approx}^k~v$), iff the following holds:

\begin{enumerate}
  \item if $k=0$, then
  \nodeLabel{u} = \nodeLabel{v}.
  \item if $k>0$, then:
  \begin{enumerate}
  \item $u~\hat{\approx}^{k-1}~v$

  \item $\forall u' \in N [(u,u') \in E \Rightarrow \exists v' \in N [(v,v') \in E,\allowbreak ~u' \hat{\approx}^{k-1} v'~{and}~\edgeLabel{u,u'} = \edgeLabel{v,v'}] ]$, and
\item
$\forall v' \in N [(v,v') \in E \Rightarrow \exists u' \in N [(u,u') \in E,\allowbreak ~v' \hat{\approx}^{k-1} u'~{and}~\edgeLabel{v,v'} = \edgeLabel{u,u'}] ]$.

\end{enumerate}
\end{enumerate}
\end{definition}

\begin{proposition}
$u~\hat{\approx}^k~v$ iff $u \approx^k v$.
\label{prop:alternative}
\end{proposition}

\begin{proof}
We first want to prove that $\approx^k \Rightarrow \hat{\approx}^k$. We will do it inductively.

(1) $k=0$. This is obvious, according to the definitions.

(2) $k>0$. Assume that for nodes $u,v \in N$, $u \approx^j v \Rightarrow u~\hat{\approx}^j~v~(0<j<k)$, we want to show
$u \approx^{j+1} v \Rightarrow u~\hat{\approx}^{j+1}~v$.
We only need to show that $u \approx^{j+1} v \Rightarrow u \approx^j v$. Since this holds, and
$u \approx^j v \Rightarrow u~\hat{\approx}^j~v$ ,
$u \approx^{j+1} v \Rightarrow u~\hat{\approx}^j~v$. Then according to Definition \ref{def:kaushik} we are done.

We then prove that $\hat{\approx}^k \Rightarrow \approx^k$. We will do it inductively.

(1) $k=0$. This is obvious, according to the definitions.

(2) $k>0$. Assume that for nodes $u,v \in N$, $u~\hat{\approx}^j~v \Rightarrow u \approx^j v~(0<j<k)$, we want to show
$u~\hat{\approx}^{j+1}~v \Rightarrow u \approx^{j+1} v$. We only need to show that \nodeLabel{u} = \nodeLabel{v}.
From $u~\hat{\approx}^{j+1}~v$, we know that $u~\hat{\approx}^{j}~v$,
therefore $u \approx^j v$. So \nodeLabel{u} = \nodeLabel{v}. Proof done.
\end{proof}

We also have another alternative definition for $k$-\emph{bisimilar} \cite{YiHSY04}:

\begin{definition}
Let \graph\ be a graph.
Let $\mathcal{I} = \{I_1, \ldots, I_n\}, n >0,$ be a set of subsets of $N$.
$\mathcal{I}$ is said to {\em partition} $G$ (or, to be a partition of $G$) if its
elements are pairwise disjoint and $N = \bigcup_{I\in\mathcal{I}} I$.
Partition  $\mathcal{I}$ is said to {\em refine} partition  $\mathcal{J}$ (or, is
a {\em refinement} of $\mathcal{J}$) if for every $I\in\mathcal{I}$ there
exists a  $J\in\mathcal{J}$ such that $I\subseteq J$.
\end{definition}

\begin{definition}
Let \graph\ be a graph and $\mathcal{I}$ and $\mathcal{J}$
be two partitions of $G$.  $\mathcal{I}$ is said to be {\em stable} with respect
to $\mathcal{J}$ if for any $I\in\mathcal{I}$, $J\in\mathcal{J}$, and edge label
$\ell$, it holds
that either $I \subseteq parents_\ell(J)$ or
$I \cap parents_\ell(J) = \emptyset$ (where $parents_\ell(J) = \{y \mid \exists x\in J (
 (y, x)\in E \textrm{~and~} edgeLabel(y, x) = \ell)\}$).

\end{definition}

\begin{definition}
Let $k\geq 0$ and \graph\ be a graph.
The $k$-partition of $G$ is defined inductively as follows:
\begin{enumerate}
\item if $k=0$, then the $k$-partition of $G$ is the set formed by partitioning $N$ by node labels.
\item if $k>0$, then the $k$-partition of $G$ is the smallest
(i.e., least cardinality)
partition $\mathcal{I}$
of $G$  such that there exists a $(k-1)$-partition
$\mathcal{J}$ of $G$ such that
$\mathcal{I}$ is a refinement of $\mathcal{J}$ and is stable with respect to $\mathcal{J}$.
\end{enumerate}

\end{definition}

\begin{definition}
\label{def:paige}
Let $k\geq 0$ and \graph\ be a graph.
Nodes $u,v \in N$ are called {\em Paige-Tarjan $k$-bisimilar} (denoted as
$u \stackrel{*}{\approx}_k v$), iff there exists an element $B$ in the
$k$-partition of $G$ such that $u\in B$ and $v\in B$.

\end{definition}

\begin{proposition}
$u \stackrel{*}{\approx}_k  v$ iff $u \approx^k v$.
\end{proposition}

\begin{proof}
We first want to prove that $\ptkbisim{k} \Rightarrow \mykbisim{k}$. We will do it inductively.

(1) $k=0$. This is obvious, since nodes are partitioned by node labels.

(2) $k>0$.
Assume that for nodes $u,v \in N$, $u \ptkbisim{j} v \Rightarrow u \mykbisim{j} v~(0<j<k)$, we want to show
$u \ptkbisim{j+1} v \Rightarrow u \mykbisim{j+1} v$.
We only need to prove point 2 and 3 of Definition \ref{def:kbisim} (\mykbisim{k}).
From $u~\ptkbisim{j+1}~v$, we know that
$\forall u' \in N [(u,u') \in E \Rightarrow \exists v' \in N [(v,v') \in E,~u' \ptkbisim{j} v'\textrm{~and~}\edgeLabel{u,u'} = \edgeLabel{v,v'}] ]$, and
$\forall v' \in N [(v,v') \in E \Rightarrow \exists u' \in N [(u,u') \in E,~v' \ptkbisim{j} u'\textrm{~and~}\edgeLabel{v,v'} = \edgeLabel{u,u'}] ]$.
Since $\ptkbisim{j} \Rightarrow \mykbisim{j}$, we have $u \mykbisim{j+1} v$.

We then prove that $\mykbisim{k} \Rightarrow \ptkbisim{k}$. We will do it inductively.

(1) $k=0$. This is obvious, since nodes are partitioned by node labels.

(2) $k>0$.
Assume that for nodes $u,v \in N$, $u \mykbisim{j} v \Rightarrow u \ptkbisim{j} v~(0<j<k)$, we want to show
$u \mykbisim{j+1} v \Rightarrow u \ptkbisim{j+1} v$.
From $u \mykbisim{j+1} v$, we know that
$\forall u' \in N [(u,u') \in E \Rightarrow \exists v' \in N [(v,v') \in E,~u' \mykbisim{j} v'\textrm{~and~}\edgeLabel{u,u'} = \edgeLabel{v,v'}] ]$, and
$\forall v' \in N [(v,v') \in E \Rightarrow \exists u' \in N [(u,u') \in E,~v' \mykbisim{j} u'\textrm{~and~}\edgeLabel{v,v'} = \edgeLabel{u,u'}] ]$.
And since $\mykbisim{j} \Rightarrow \ptkbisim{j}$, we know that all children of $u,v$ who have the same edge label belong to the same partition.
This fulfills the stable condition.
From Proposition \ref{prop:1} we know that \mykbisim{j+1} is a refinement of \mykbisim{j}.
We then have $u \ptkbisim{j+1} v$.
\end{proof}

\subsection{Partition splitting stop condition}
\label{sec:stopcondition}

\begin{proposition}
If $\mykbisim{j}=\mykbisim{j+1}$, then $\mykbisim{j}=\mykbisim{j'}~(\forall j' \ge j)$.
\label{prop:stay}
\end{proposition}
\begin{proof}
Since $\mykbisim{j}=\mykbisim{j+1}$, $\forall u \in N$, we could assign $pId_j(u)=pId_{j+1}(u)$.
Then, according to Definition \ref{def:sig} (signature) and Proposition \ref{prop:equal}, it holds that $pId_{j+2}(u) = pId_{j+1}(u)$,
and the same applies for any further $j' \ge j$.
\end{proof}

\begin{proposition}
The $j$ in Proposition \ref{prop:stay} always exists, and its upper bound is $|N|$ (number of nodes).
\label{prop:upperbound}
\end{proposition}
\begin{proof}

From Proposition \ref{prop:1}, we know that $\forall u,v \in N$,
if $u \mykbisim{j+1} v$, then $u \mykbisim{j} v$, which is equivalent of saying partitions will either split or stay the same.
If they stay for one time, they will stay forever (Proposition \ref{prop:stay}).
Otherwise, $G$ has to at least split one of its
partition blocks for each \mykbisim{i} where $i \le j$ , in which case $j$ reach the
upper bound $|N|$.
\end{proof}

From Proposition \ref{prop:upperbound}, we know that there is an upper bound for the number of iterations in Algorithm \ref{algo:build_k-bisims}.
If this upper bound is smaller than the user input $k$, algorithm can terminate earlier.
Since the partition blocks will either split or remain the same, the number of partition blocks will either increase or remain.
Therefore, by simply checking if two consecutive iterations produce the same number of partition blocks,
we could decide whether the computation should stop.

\subsection{Connection between localized bisimulation and full bisimulation}
We observe the following useful connection between localized and full
bisimulation.

\begin{definition}
Let $k\geq 0$ and \graph\ be a graph.
Nodes $u,v \in N$ are called {\em bisimilar} (denoted as
$u \approx v$), iff the following holds:
\begin{enumerate}
  \item
  \nodeLabel{u} = \nodeLabel{v},
  \item
  $\forall u' \in N [(u,u') \in E \Rightarrow \exists v' \in N [(v,v') \in E,~u' \approx v'~and~
  \edgeLabel{u,u'} = \edgeLabel{v,v'}] ]$, and
  \item
  $\forall v' \in N [(v,v') \in E \Rightarrow \exists u' \in N [(u,u') \in E,~v' \approx u'~and~
  \edgeLabel{v,v'} = \edgeLabel{u,u'}] ]$.
\end{enumerate}

\end{definition}

\begin{proposition}
Let \graph\ be a graph.
There exists a $k\ge 0$ such that for any
$u,v \in N$ it holds that
$u \approx_k v$  iff $u \approx v$.
\end{proposition}

\begin{proof}
First we want to show $u \approx_k v \Rightarrow u \approx v$.
From Proposition \ref{prop:upperbound}, we know that $k$ has an upper bound $|N|$. Here we set $k$ to $|N|$,
which means that $\approx_k = \approx_{k+1}$. Then according to the definition, in iteration $k+1$, for $u \approx_{k+1} v$, we have:
\begin{enumerate}
  \item
  \nodeLabel{u} = \nodeLabel{v},
  \item
  $\forall u' \in N [(u,u') \in E \Rightarrow \exists v' \in N [(v,v') \in E,~u' \approx^{k} v'~and~
  \edgeLabel{u,u'} = \edgeLabel{v,v'}] ]$, and
  \item
  $\forall v' \in N [(v,v') \in E \Rightarrow \exists u' \in N [(u,u') \in E,~v' \approx^{k} u'~and~
  \edgeLabel{v,v'} = \edgeLabel{u,u'}] ]$.
\end{enumerate}
Since $\approx_k = \approx_{k+1}$, we can replace $\approx_k$ with $\approx_{k+1}$, then the relationship $\approx_{k+1}$ has the same definition as $\approx$. So that $\approx_{k+1} = \approx$.

Then we want to show that $u \approx v \Rightarrow u \approx_k v$. We will do it inductively.
\begin{enumerate}
  \item $k=0$. This is obvious.
  \item $k>0$. Assume that this holds for $j-1$, we want to show that this also holds for $j$. Let $u \approx v$, we want to show that $u \approx_j v$. According to the definition, we want to have for all outgoing edges $(u,u') \in E$, there exists some edge $(v,v') \in E$, such that $u' \approx^{j-1} v'$ and \edgeLabel{u,u'} = \edgeLabel{v,v'}, and vice versa. Because of $u \approx v$, we already have $u' \approx v'$; and because of $u \approx v \Rightarrow u \approx_{j-1} v$, we have $u' \approx_{j-1} v'$. Then all the requirements for $u \approx_j v$ are fulfilled. So $\approx \Rightarrow \approx_k$.
\end{enumerate}\end{proof}

%\newpage
%\section{Experiment result numbers}

%\begin{table*}[hbp]
\begin{sidewaystable*}[h]
\centering
\caption{Experiment results of \construct\ for real and synthetic datasets}
\label{table:big_table}%
%\singlespacing
\scriptsize
\tabcolsep=0.1cm
% Table generated by Excel2LaTeX from sheet 'real'
%\scalebox{0.5}{
%\resizebox{17cm}{!}{
\resizebox{\textwidth}{!}{
% Table generated by Excel2LaTeX from sheet 'ALL'
\begin{tabular}{rrrrrrrrrrrr}
\hline
Data Set & Measurement & Iteration 1 & Iteration 2 & Iteration 3 & Iteration 4 & Iteration 5 & Iteration 6 & Iteration 7 & Iteration 8 & Iteration 9 & Iteration 10 \\
\hline
\multicolumn{1}{c}{\multirow{8}[0]{*}{Jamendo}} & Partition Count & 43    & 199   & 297   & 310   & 310   & 310   & 310   & 310   & 310   & 310 \\
\multicolumn{1}{c}{} & Preparation Time (s) & 0.88  & 0.63  & 0.68  & 0.69  & 0.68  & 0.65  & 0.69  & 0.66  & 0.65  & 0.65 \\
\multicolumn{1}{c}{} & Constructing Time (s) & 1.78  & 2.05  & 2.16  & 2.17  & 2.18  & 2.17  & 2.14  & 2.38  & 2.42  & 2.40 \\
\multicolumn{1}{c}{} & Table Read (byte) & 111,149,056 & 75,497,472 & 75,497,472 & 75,497,472 & 75,497,472 & 75,497,472 & 75,497,472 & 75,497,472 & 75,497,472 & 75,497,472 \\
\multicolumn{1}{c}{} & Table Write (byte) & 113,246,208 & 77,594,624 & 77,594,624 & 77,594,624 & 77,594,624 & 77,594,624 & 77,594,624 & 77,594,624 & 77,594,624 & 77,594,624 \\
\multicolumn{1}{c}{} & S Read (byte) & 8,192 & 0     & 0     & 0     & 0     & 0     & 0     & 0     & 0     & 0 \\
\multicolumn{1}{c}{} & S Write (byte) & 4,096 & 40,960 & 69,632 & 98,304 & 122,880 & 163,840 & 176,128 & 237,568 & 241,664 & 278,528 \\
\multicolumn{1}{c}{} & Max Signature Length & 21    & 23    & 23    & 23    & 23    & 23    & 23    & 23    & 23    & 23 \\
\hline
\multicolumn{1}{c}{\multirow{8}[0]{*}{LinkedMDB}} & Partition Count & 8,460 & 38,291 & 71,161 & 85,327 & 85,660 & 85,692 & 85,704 & 85,707 & 85,709 & 85,711 \\
\multicolumn{1}{c}{} & Preparation Time (s) & 5.78  & 4.88  & 4.94  & 4.77  & 4.86  & 5.87  & 4.91  & 4.90  & 4.79  & 4.80 \\
\multicolumn{1}{c}{} & Constructing Time (s) & 12.29 & 13.58 & 13.73 & 14.49 & 14.00 & 14.56 & 15.58 & 14.46 & 14.51 & 16.05 \\
\multicolumn{1}{c}{} & Table Read (byte) & 731,906,048 & 597,688,320 & 597,688,320 & 597,688,320 & 597,688,320 & 597,688,320 & 597,688,320 & 597,688,320 & 597,688,320 & 597,688,320 \\
\multicolumn{1}{c}{} & Table Write (byte) & 884,998,144 & 752,877,568 & 752,877,568 & 752,877,568 & 752,877,568 & 752,877,568 & 752,877,568 & 752,877,568 & 752,877,568 & 752,877,568 \\
\multicolumn{1}{c}{} & S Read (byte) & 8,192 & 0     & 0     & 0     & 0     & 0     & 0     & 0     & 0     & 0 \\
\multicolumn{1}{c}{} & S Write (byte) & 1,982,464 & 7,389,184 & 16,576,512 & 24,403,968 & 33,886,208 & 45,350,912 & 57,716,736 & 68,988,928 & 76,664,832 & 80,326,656 \\
\multicolumn{1}{c}{} & Max Signature Length & 63    & 179   & 203   & 229   & 243   & 243   & 243   & 243   & 243   & 243 \\
\hline
\multicolumn{1}{c}{\multirow{8}[0]{*}{DBLP}} & Partition Count & 246   & 9,073 & 11,130 & 11,189 & 11,189 & 11,189 & 11,189 & 11,189 & 11,189 & 11,189 \\
\multicolumn{1}{c}{} & Preparation Time (s) & 59.64 & 43.41 & 44.10 & 44.50 & 45.34 & 46.33 & 46.58 & 48.03 & 46.88 & 46.65 \\
\multicolumn{1}{c}{} & Constructing Time (s) & 98.53 & 112.46 & 114.80 & 117.44 & 118.09 & 116.79 & 117.96 & 117.52 & 119.69 & 118.03 \\
\multicolumn{1}{c}{} & Table Read (byte) & 8,044,675,072 & 5,731,516,416 & 5,733,613,568 & 5,733,613,568 & 5,733,613,568 & 5,733,613,568 & 5,733,613,568 & 5,733,613,568 & 5,733,613,568 & 5,733,613,568 \\
\multicolumn{1}{c}{} & Table Write (byte) & 9,353,297,920 & 7,092,568,064 & 7,096,762,368 & 7,096,762,368 & 7,096,762,368 & 7,096,762,368 & 7,096,762,368 & 7,096,762,368 & 7,096,762,368 & 7,096,762,368 \\
\multicolumn{1}{c}{} & S Read (byte) & 8,192 & 0     & 0     & 0     & 0     & 0     & 0     & 0     & 0     & 0 \\
\multicolumn{1}{c}{} & S Write (byte) & 53,248 & 2,854,912 & 3,956,736 & 4,546,560 & 5,300,224 & 7,876,608 & 10,571,776 & 13,393,920 & 15,962,112 & 18,608,128 \\
\multicolumn{1}{c}{} & Max Signature Length & 37    & 99    & 723   & 745   & 745   & 745   & 745   & 745   & 745   & 745 \\
\hline
\multicolumn{1}{c}{\multirow{8}[0]{*}{Wikilinks}} & Partition Count & 2     & 4     & 14    & 327   & 928,765 & 2,992,705 & 3,596,837 & 3,604,409 & 3,605,063 & 3,605,151 \\
\multicolumn{1}{c}{} & Preparation Time (s) & 137.65 & 108.19 & 107.01 & 106.82 & 108.28 & 117.37 & 115.32 & 115.57 & 116.89 & 116.30 \\
\multicolumn{1}{c}{} & Constructing Time (s) & 17.92 & 19.71 & 20.31 & 28.88 & 62.53 & 193.49 & 436.94 & 441.90 & 614.60 & 632.42 \\
\multicolumn{1}{c}{} & Table Read (byte) & 15,065,939,968 & 10,798,235,648 & 10,831,790,080 & 10,888,413,184 & 11,156,848,640 & 12,318,670,848 & 12,580,814,848 & 12,593,397,760 & 12,595,494,912 & 12,595,494,912 \\
\multicolumn{1}{c}{} & Table Write (byte) & 17,205,035,008 & 12,939,427,840 & 13,006,536,704 & 13,119,782,912 & 13,656,653,824 & 15,980,298,240 & 16,504,586,240 & 16,529,752,064 & 16,533,946,368 & 16,533,946,368 \\
\multicolumn{1}{c}{} & S Read (byte) & 8,192 & 0     & 0     & 0     & 24,797,184 & 8,697,421,824 & 12,882,042,880 & 16,277,565,440 & 17,842,032,640 & 19,003,453,440 \\
\multicolumn{1}{c}{} & S Write (byte) & 4,096 & 4,096 & 4,096 & 36,864 & 205,180,928 & 8,431,570,944 & 12,294,479,872 & 15,117,230,080 & 15,919,968,256 & 16,244,105,216 \\
\multicolumn{1}{c}{} & Max Signature Length & 3     & 5     & 9     & 19    & 129   & 6,817 & 8,349 & 9,363 & 9,421 & 9,425 \\
\hline
\multicolumn{1}{c}{\multirow{8}[0]{*}{Dbpedia}} & Partition Count & 362,128 & 2,357,366 & 3,239,710 & 3,273,445 & 3,281,100 & 3,299,007 & 3,343,927 & 3,401,435 & 3,436,428 & 3,450,357 \\
\multicolumn{1}{c}{} & Preparation Time (s) & 146.29 & 108.95 & 113.45 & 116.87 & 114.84 & 116.22 & 116.30 & 116.28 & 114.72 & 119.99 \\
\multicolumn{1}{c}{} & Constructing Time (s) & 213.61 & 366.58 & 466.13 & 585.40 & 632.35 & 664.99 & 679.87 & 763.97 & 863.79 & 1,117.25 \\
\multicolumn{1}{c}{} & Table Read (byte) & 16,760,438,784 & 12,366,905,344 & 12,574,523,392 & 12,595,494,912 & 12,605,980,672 & 12,616,466,432 & 12,629,049,344 & 12,639,535,104 & 12,643,729,408 & 12,643,729,408 \\
\multicolumn{1}{c}{} & Table Write (byte) & 19,295,895,552 & 15,453,913,088 & 15,869,149,184 & 15,911,092,224 & 15,932,063,744 & 15,953,035,264 & 15,978,201,088 & 15,999,172,608 & 16,007,561,216 & 16,007,561,216 \\
\multicolumn{1}{c}{} & S Read (byte) & 8,192 & 3,870,638,080 & 5,215,023,104 & 5,915,021,312 & 6,404,620,288 & 7,598,112,768 & 8,796,708,864 & 9,225,072,640 & 10,492,932,096 & 11,405,717,504 \\
\multicolumn{1}{c}{} & S Write (byte) & 123,658,240 & 4,553,515,008 & 5,857,165,312 & 6,507,692,032 & 6,952,050,688 & 8,115,949,568 & 9,237,364,736 & 9,629,892,608 & 10,667,528,192 & 11,225,329,664 \\
\multicolumn{1}{c}{} & Max Signature Length & 1,501 & 5,109 & 7,687 & 8,179 & 8,213 & 8,215 & 8,269 & 8,269 & 8,269 & 8,269 \\
\hline
\multicolumn{1}{c}{\multirow{8}[1]{*}{Twitter}} & Partition Count & 2     & 4     & 16    & 1,463 & 14,251,228 & 35,729,811 & 36,178,375 & 36,192,245 & 36,192,750 & 36,192,805 \\
\multicolumn{1}{c}{} & Preparation Time (s) & 4,980.11 & 4,221.10 & 4,226.36 & 4,310.65 & 4,290.77 & 4,577.37 & 4,554.91 & 4,446.23 & 4,410.29 & 4,422.27 \\
\multicolumn{1}{c}{} & Constructing Time (s) & 170.97 & 215.58 & 260.66 & 362.50 & 1,795.55 & 3,881.94 & 3,876.89 & 3,984.64 & 4,051.14 & 5,012.17 \\
\multicolumn{1}{c}{} & Table Read (byte) & 168,455,831,552 & 120,275,861,504 & 120,674,320,384 & 121,194,414,080 & 124,528,885,760 & 141,601,800,192 & 142,751,039,488 & 142,753,136,640 & 142,753,136,640 & 142,753,136,640 \\
\multicolumn{1}{c}{} & Table Write (byte) & 192,552,108,032 & 144,531,521,536 & 145,328,439,296 & 146,368,626,688 & 153,037,570,048 & 187,183,398,912 & 189,481,877,504 & 189,486,071,808 & 189,486,071,808 & 189,486,071,808 \\
\multicolumn{1}{c}{} & S Read (byte) & 8,192 & 0     & 0     & 0     & 33,206,579,200 & 130,105,450,496 & 116,607,520,768 & 137,478,197,248 & 151,362,093,056 & 162,154,037,248 \\
\multicolumn{1}{c}{} & S Write (byte) & 4,096 & 4,096 & 4,096 & 155,648 & 34,853,953,536 & 115,356,168,192 & 110,612,774,912 & 119,332,966,400 & 121,634,852,864 & 123,151,667,200 \\
\multicolumn{1}{c}{} & Max Signature Length & 3     & 5     & 9     & 33    & 1,373 & 4,354,479 & 5,840,263 & 5,848,053 & 5,848,119 & 5,848,119 \\
\hline
\multicolumn{1}{c}{\multirow{8}[2]{*}{SP2B}} & Partition Count & 728   & 219,581 & 459,986 & 467,369 & 467,369 & 467,369 & 467,369 & 467,369 & 467,369 & 467,369 \\
\multicolumn{1}{c}{} & Preparation Time (s) & 1,238.28 & 859.67 & 842.68 & 850.36 & 851.59 & 831.17 & 841.14 & 877.76 & 847.49 & 854.87 \\
\multicolumn{1}{c}{} & Constructing Time (s) & 1,392.42 & 1,670.65 & 1,824.11 & 1,929.69 & 2,066.06 & 2,152.50 & 2,265.44 & 2,248.45 & 2,226.88 & 2,337.82 \\
\multicolumn{1}{c}{} & Table Read (byte) & 105,736,306,688 & 68,232,937,472 & 68,232,937,472 & 68,232,937,472 & 68,232,937,472 & 68,232,937,472 & 68,232,937,472 & 68,232,937,472 & 68,232,937,472 & 68,232,937,472 \\
\multicolumn{1}{c}{} & Table Write (byte) & 120,431,050,752 & 82,931,875,840 & 82,931,875,840 & 82,931,875,840 & 82,931,875,840 & 82,931,875,840 & 82,931,875,840 & 82,931,875,840 & 82,931,875,840 & 82,931,875,840 \\
\multicolumn{1}{c}{} & S Read (byte) & 8,192 & 0     & 0     & 2,285,568 & 26,083,328 & 221,638,656 & 390,049,792 & 425,611,264 & 443,654,144 & 446,136,320 \\
\multicolumn{1}{c}{} & S Write (byte) & 118,784 & 62,963,712 & 97,890,304 & 136,470,528 & 196,829,184 & 387,956,736 & 495,534,080 & 514,424,832 & 523,583,488 & 523,796,480 \\
\multicolumn{1}{c}{} & Max Signature Length & 109   & 109   & 109   & 109   & 109   & 109   & 109   & 109   & 109   & 109 \\
\hline
\multicolumn{1}{c}{\multirow{8}[2]{*}{BSBM}} & Partition Count & 50    & 510   & 511   & 512   & 512   & 512   & 512   & 512   & 512   & 512 \\
\multicolumn{1}{c}{} & Preparation Time (s) & 38.54 & 28.52 & 28.14 & 27.88 & 28.07 & 27.92 & 27.96 & 27.79 & 28.03 & 27.86 \\
\multicolumn{1}{c}{} & Constructing Time (s) & 59.91 & 61.32 & 59.62 & 63.29 & 63.51 & 64.26 & 64.11 & 65.09 & 64.40 & 65.31 \\
\multicolumn{1}{c}{} & Table Read (byte) & 5,179,965,440 & 3,764,387,840 & 3,764,387,840 & 3,764,387,840 & 3,764,387,840 & 3,764,387,840 & 3,764,387,840 & 3,764,387,840 & 3,764,387,840 & 3,764,387,840 \\
\multicolumn{1}{c}{} & Table Write (byte) & 6,228,541,440 & 4,819,255,296 & 4,819,255,296 & 4,819,255,296 & 4,819,255,296 & 4,819,255,296 & 4,819,255,296 & 4,819,255,296 & 4,819,255,296 & 4,819,255,296 \\
\multicolumn{1}{c}{} & S Read (byte) & 8,192 & 0     & 0     & 0     & 0     & 0     & 0     & 0     & 0     & 0 \\
\multicolumn{1}{c}{} & S Write (byte) & 16,384 & 106,496 & 110,592 & 167,936 & 270,336 & 442,368 & 405,504 & 585,728 & 573,440 & 499,712 \\
\multicolumn{1}{c}{} & Max Signature Length & 35    & 37    & 37    & 37    & 37    & 37    & 37    & 37    & 37    & 37 \\
\hline
\end{tabular}%
}
%\end{table*}%
\end{sidewaystable*}

\begin{table*}
\centering
\caption{Sum-up of experiment results from Table~\ref{table:big_table}}
% Table generated by Excel2LaTeX from sheet 'sum-up table'
\resizebox{\textwidth}{!}{
\begin{tabular}{rrrrrrrrr}
\hline
      & Jamendo & LinkedMDB & DBLP  & Wikilinks & Dbpedia & Twitter & SP2B  & BSBM \\
\hline
Partition count / Node Count & 0.064\% & 3.677\% & 0.049\% & 63.127\% & 8.935\% & 86.893\% & 0.166\% & 0.006\% \\
Elapsed Time (s) & 28.72712 & 193.748 & 1622.765 & 3618.1302 & 7537.846 & 68052.12 & 29009.04 & 921.5325 \\
Overall I/O (byte) & 1.6E+09 & 1.42E+10 & 1.33E+11 & 4.164E+11 & 4.34E+11 & 4.451E+12 & 1.59E+12 & 8.87E+10 \\
Max Signature Length / Node Count & 0.00473\% & 0.01043\% & 0.00324\% & 0.16503\% & 0.02141\% & 14.04035\% & 0.00004\% & 0.00042\% \\
\hline
\end{tabular}%
}
\end{table*}

%\section{When will partition blocks stop splitting?}
%
%When $k$ reaches a certain point, the partition blocks will stop splitting and remain
%stable. In this case we do not need to do more rounds of computation. So we need
%a method to detect if the partition remains the same after each iteration.
%
%The first intuition which comes to mind is to detect the change in every node's
%signature. If each node $u$'s signature remains the same in one iteration, then
%$pId(u)$ also remains the same, so that there is no change in the
%partition. Interestingly, this does not hold for the other way round, meaning
%that when partitions are stable, signatures of the nodes could change (forever).
%This is due to the fact that we do not make any stronger assumption for \s.
%
%This can be illustrated by an example (Figure \ref{fig:example_forever}). In
%this example we assume that our partition numbering function is not independent
%across iterations (different iterations share the same \s). Then in this
%example, node 1 and 2 both have an edge pointing to each other, which creates an
%infinite loop for signature changes: whenever one node's $pId$ changes, the
%other's signature will change and vice versa.
%
%\begin{figure}[htbp]
%\centering
%\include{figures/example_forever}
%\caption{Example of infinite signature changes}
%\label{fig:example_forever}
%\end{figure}
%
%\section{Alternative definitions for localized bisimulation}

\end{appendix}

}{
%do nothing in shorter version
}

\end{document}